\documentclass[singlecolumn,prx,longbibliography,nofootinbib, superscriptaddress]{revtex4-2}

\usepackage[T1]{fontenc}
\usepackage[a4paper,margin=1in]{geometry}
\usepackage{amsmath,amssymb,amsthm,mathtools,bm,mathrsfs}
\usepackage{physics}
\usepackage{xcolor}
\usepackage{enumitem}
\usepackage{graphicx}
\usepackage{tikz}
\usetikzlibrary{positioning}

\usepackage{times}
\usepackage[colorlinks=true,linkcolor=teal,citecolor=teal,urlcolor=teal]{hyperref}

\newtheorem{theorem}{Theorem}[section]
\newtheorem{lemma}[theorem]{Lemma}
\newtheorem{corollary}[theorem]{Corollary}
\newtheorem{proposition}[theorem]{Proposition}
\newtheorem{remark}[theorem]{Remark}
\newtheorem{example}[theorem]{Example}
\newtheorem{definition}[theorem]{Definition}

\newcommand{\E}{\mathbb{E}}
\newcommand{\one}{\mathbf{1}}

\newcommand{\trnorm}[1]{\left\lVert #1 \right\rVert_1}
\newcommand{\dist}{\operatorname{dist}}
\newcommand{\id}{\operatorname{id}}
\newcommand{\HS}[1]{\lVert #1 \rVert_2}

\begin{document}

\title{Approximation theorems for fermionic Gaussian states}


\author{Amir-Reza Negari}
\affiliation{Perimeter Institute for Theoretical Physics, Waterloo, Ontario N2L 2Y5, Canada}
 \affiliation{Institute for Quantum Computing, University of Waterloo, Waterloo, Ontario N2L 3G1, Canada}

\author{Farzin Salek}
\affiliation{Perimeter Institute for Theoretical Physics, Waterloo, Ontario N2L 2Y5, Canada}
\affiliation{Institute for Quantum Computing, University of Waterloo, Waterloo, Ontario N2L 3G1, Canada}
\affiliation{Dahlem Center for Complex Quantum Systems, Freie Universit\"at Berlin, 14195 Berlin, Germany}

\author{Zolt\'an Zimbor\'as}
\affiliation{Department of Physics, University of Helsinki, P.O. Box 43, FI-00014 Helsinki, Finland}
\affiliation{HUN-REN Wigner Research Centre for Physics, Budapest, Hungary}
\affiliation{Algorithmiq Ltd, Kanavakatu 3C, 00160 Helsinki, Finland}

\author{Aram Harrow}
\affiliation{Center for Theoretical Physics -- a Leinweber Institute, Massachusetts Institute of Technology, Cambridge, Massachusetts 02139, USA}

\author{Patrick Hayden}
\affiliation{Leinweber Institute for Theoretical Physics, Stanford University, Stanford, California 94305, USA}
\affiliation{Google DeepMind}

\author{Jens Eisert}
\affiliation{Dahlem Center for Complex Quantum Systems, Freie Universit\"at Berlin, 14195 Berlin, Germany}
\affiliation{Helmholtz-Zentrum Berlin f\"ur Materialien und Energie, 14109 Berlin, Germany}

\begin{abstract}
A recurring principle in many-body physics is that sufficiently distributed interactions or correlations lead to an effectively mean-field description. For fermionic Gaussian states, we establish a quantitative version of this principle directly at the level of Majorana covariance matrices. The main observation is that the admissibility condition for fermionic covariance matrices imposes a quantitative monogamy constraint on two-point correlations: a fixed region has only a bounded covariance budget to distribute among many disjoint regions. We use this principle in three directions.  First, for quadratic fermionic Hamiltonians on finite simple graphs of high degree, we obtain product-state approximations to both the ground-state energy and the finite-temperature free energy, with error of order \(D^{-1/2}\) for \(D\)-regular graphs.  Second, we prove a finite
fermionic Gaussian de Finetti theorem: a \(k\)-copy Gaussian state that is
\(n\)-exchangeable under graded fermionic copy permutations is
\(O(k/n)\)-close in trace norm to the product of its one-copy marginal.  In particular, infinite exchangeability within the Gaussian class implies exact product structure. Third, for spatial fermionic Gaussian
states with exponentially clustering mutual information, we prove that the
conditional mutual information across a separating buffer decays as
\(O(r^{-1})\), and hence, by recoverability, the state admits an
\(O(r^{-1/2})\) trace-norm approximation by a recovered state.  These
results show that product and de Finetti approximations and recoverability
estimates for fermionic Gaussian systems can all be derived directly from
covariance-level structure.
\end{abstract}
\maketitle


\section{Introduction}

Fermionic Gaussian states, also known as quasifree or free-fermionic
states, constitute a central class of quantum states in many-body physics.
They arise naturally as ground states and thermal states of quadratic, and
hence noninteracting, fermionic Hamiltonians.  Remarkably, their reduced
states, entropies, and correlation functions can all be efficiently determined
from one-particle data, making them one of the paradigmatic exactly tractable
classes of quantum many-body states.  In a Majorana representation, these data
are encoded in the Majorana covariance matrix, denoted throughout this work
by
\(M\); equivalently, one may use the one-particle Green operator
\(G=(\one+iM)/2\).  The covariance description is classical in the theory of
quasifree states and is also central in fermionic linear optics and in the
analysis of free lattice models
\citep{ArakiQuasifree,BravyiFLO,TerhalDiVincenzo,Peschel2003,PeschelEisler2009}.
At the same time, fermionic Gaussian states are not merely a technical toy
model.  They include physically important examples with Fermi surfaces and
characteristic entanglement behaviour, including logarithmic violations of
entropic area laws
\citep{WolfAreaViolation,GioevKlich,EisertCramerPlenio,SwingleFermiSurface}.
They therefore provide a useful setting in which to test general ideas about
correlation, locality, exchangeability, and approximation in many-body quantum
systems.

One particularly important idea is mean-field theory.  Its basic
premise is that, when the influence on a local degree of freedom is distributed
among sufficiently many other degrees of freedom, the correlations carried by
any individual interaction become weak and a low-complexity variational
description becomes accurate.  In fermionic many-body theory, Hartree--Fock
theory is the canonical realization of this premise: one optimizes over Slater
determinants, or more generally over quasifree states when pairing is allowed,
thereby replacing an interacting problem by a self-consistent one-particle
description.  The setting of the present work is narrower but also more
quantitative.  We remain within the Gaussian class and ask when high
connectivity, symmetric extendibility, or spatial separation rigorously forces
an approximation by an even simpler Gaussian state, such as a product of local
marginals.  Thus, our product-state bounds should be understood as rigorous
Gaussian mean-field statements, closely related in spirit to Hartree--Fock but
not as an analysis of the full Hartree--Fock variational problem for interacting
fermions.

Motivated by this line of thought, the guiding question of this work is
when a fermionic Gaussian state can, for a given purpose, be replaced by a
simpler state.  We consider three precise versions of this question.  The first
is variational: when, and under what conditions, are the energy and free energy
of a quadratic fermionic Hamiltonian well approximated by product states?  The
second concerns exchangeability: when does the existence of many symmetric
extensions force a Gaussian state to be close to an independent product state?
The third is spatial and information-theoretic: when does a buffer region make
two distant regions approximately conditionally independent, so that the state
is approximately recoverable from one side of the buffer?  Although these
questions are usually treated by different methods, we show that in the
fermionic Gaussian setting they are governed by a common covariance-level
mechanism.

The starting point is the admissibility condition for Majorana covariance
matrices.  Informally, not every real antisymmetric matrix can be the
covariance matrix of a fermionic state: it must satisfy a positivity condition,
equivalently its singular values must be bounded by one.  This elementary
constraint has strong consequences for the way two-point correlations can be
distributed.  If a region \(A\) is correlated with many disjoint regions
\(B_1,\ldots,B_k\), then the covariance blocks from \(A\) to these regions
cannot all be large.  We refer to this as covariance monogamy.  The main point
of {this work
is that covariance monogamy, combined either with high degree,
exchangeability, or spatial separation, leads to quantitative approximation
theorems.

The first application concerns product-state approximations for high-degree
Hamiltonians.  For general two-local qudit Hamiltonians, Brand\~ao and Harrow
proved that high degree in the interaction graph forces the ground-state energy
to be well approximated by product states
\citep{BrandaoHarrowGround,BrandaoHarrowStates}.  Their result is based on
information-theoretic monogamy and applies far beyond the fermionic setting.
For quadratic fermionic Hamiltonians, however, the energy of an edge is
determined only by the two-point covariance block across that edge.  This makes
it possible to bypass dimension-dependent many-body approximation theorems and
work directly with covariance matrices.  In this way we obtain an
\(O(D^{-1/2})\) product-state approximation for quadratic Hamiltonians on
finite simple \(D\)-regular graphs.  If each site contains \(m\) fermionic
modes and the edge
terms are normalized in the high-degree convention, then there exists a product
state \(\varphi\) such that
\begin{equation}
        \tr(H\varphi)
        \le
        e_0(H)+\sqrt{\frac{2m}{D}}.
\end{equation}
The same argument gives the corresponding finite-temperature statement:
for a quadratic Hamiltonian, the product of the one-site marginals of the Gibbs
state has free energy within \(O(D^{-1/2})\) of the Gibbs variational optimum.
This result directly develops and substantially extends the program initiated in our earlier work
\citep{KrumnowZimborasEisertGroundApprox}.  That work established rigorous
Gaussian product-state approximations for general interacting fermionic
Hamiltonians from the weighted interaction graph and interpreted the underlying
correlation constraint as a non-symmetric fermionic de Finetti phenomenon.  The
present high-degree theorem retains the central idea that a local fermionic
degree of freedom cannot sustain strong correlations with many neighbours, but
specializes it to quadratic Hamiltonians, where every edge energy is exactly
determined by a two-point covariance block.  This extra structure replaces the
general many-body approximation argument by a direct covariance-budget estimate
and yields the sharper \(D^{-1/2}\) scaling.  The finite-temperature extension,
the exchangeable Gaussian de Finetti theorem below, and the approximate-Markov
result have no counterparts in Ref.~\citep{KrumnowZimborasEisertGroundApprox}.

The second application is a fermionic Gaussian de Finetti theorem.
Classical and quantum de Finetti theorems make precise the link between
extendibility and mean-field structure: permutation symmetry together with
extensions to arbitrarily many copies forces a state to be a mixture of
independent and identically distributed product states, while finite
extendibility yields quantitative approximations of this form.  This principle
has become a fundamental tool in quantum information theory, entanglement
theory, and rigorous mean-field analysis
\citep{Sturmer,HudsonMoody,ChristandlKoenigMitchisonRenner,RennerNatPhys,
BrandaoHarrowLocalMeasurements,LiSmith,RougerieDeFinetti}.  Fermionic systems
require additional care, because tensor-product language must be reconciled
with mode decompositions, anticommutation relations, and parity
superselection \citep{Negari_2026}.  These issues lead to several reasonable notions of fermionic
separability and mode entanglement
\citep{SchuchVerstraeteCiracFermions,BanulsCiracWolfFermionicEntanglement,
SpeeSchwaigerGiedkeKraus}.  They also motivate fermionic versions of de Finetti
theorems, such as the finite Majorana de Finetti theorem of Krumnow,
Zimbor\'as, and Eisert \citep{KrumnowZimborasEisert}.  Complementary fermionic
Gaussian de Finetti theorems based on Gaussian symmetry in the
matchgate-replica representation have recently been obtained, including an
exponential almost-i.i.d. version
\citep{SierantTurkeshiTarabungaMatchgateCommutant,
BurkatStudzinskiStrelchukExponentialDeFinetti}.

That said, our de Finetti result concerns a different setting.  It is restricted to
joint states that are themselves fermionic Gaussian and are extendible under
graded fermionic copy permutations, but within that setting it is completely
covariance-theoretic.  If a \(k\)-copy fermionic Gaussian state admits an
extension to \(n\) copies invariant under graded fermionic copy permutations,
then this symmetry forces all off-diagonal copy-copy covariance blocks to be
identical.
Admissibility of the \(n\)-copy covariance matrix then forces this common block
to be of order \(1/n\).  Converting the resulting covariance estimate into
trace norm, using perturbation bounds for free-fermionic states
\citep{BittelMeleEisertLeoneTraceDistance}, gives
\begin{equation}
        \left\|
        \rho_{A_1\cdots A_k}-\rho_A^{\otimes k}
        \right\|_1
        \le
        \frac{4m_A(k-1)}{n},
\end{equation}
where \(m_A\) is the number of fermionic modes in one copy.  Thus the
approximating de Finetti measure may be chosen to be the point mass at the
one-copy marginal.  In particular, an infinitely exchangeable fermionic
Gaussian state is exactly a product state.

The de Finetti result should also be compared with the recent work of
Negari and Salek \citep{NegariSalekExtendibility}.  There, the focus was on
two-sided \((k_1,k_2)\)-extendibility of bipartite fermionic Gaussian states,
its covariance-matrix characterization, and its consequences for approximation
by fermionic separable states.  The present work uses the same
covariance-matrix viewpoint and notation, but addresses a different
exchangeability problem: rather than approximating a bipartite extendible state
by the convex set of separable states, we study graded-permutation-invariant
extensions across many identical copies and prove closeness to the specific
product of the one-copy marginal.  Thus, the finite de Finetti theorem here is
an i.i.d. statement, with an explicit \(O(k/n)\) trace-norm bound, rather than a
separability statement.

Throughout this work, product and separability statements are understood with
respect to a fixed partition of fermionic modes, and all Gaussian states are
taken to be even.  For Gaussian states this convention has a simple covariance
interpretation: the product of the Gaussian marginals on disjoint regions is
obtained by deleting the off-diagonal covariance blocks between the regions.
This is the convention used in \citep{NegariSalekExtendibility}, and it is
consistent with the mode-based treatment of fermionic entanglement under parity
superselection
\citep{SchuchVerstraeteCiracFermions,BanulsCiracWolfFermionicEntanglement,
SpeeSchwaigerGiedkeKraus}.  The preliminaries section below recalls the precise
operator-algebraic and covariance-matrix conventions used in the proofs.

The third application concerns conditional mutual information and approximate
Markov structure.  For a tripartite state on \(A,B,C\), the conditional mutual
information
\begin{equation}
        I(A:C\mid B)
        =
        S(A,B)+S(B,C)-S(B)-S(A,B,C)
\end{equation}
measures the part of the correlation between \(A\) and \(C\) that remains after
the degrees of freedom in \(B\) have been retained.  Exact vanishing of this
quantity characterizes quantum Markov chains
\citep{HaydenJozsaPetzWinter,IbinsonLindenWinter}; quantitatively, the
Fawzi--Renner recoverability theorem and its refinements show that small
conditional mutual information implies approximate recovery by a quantum
channel acting on the conditioning system
\citep{FawziRenner,SutterTomamichelHarrow,JungeRennerSutterWildeWinter,
WildeRecoverability}.

Approximate Markov structure has become an important locality concept in
quantum many-body theory \citep{,negari2026criticalnonequilibriumphasesnoisy,kumar2026unlearnablephasesmatter}.  It is closely related to Gibbs-state locality,
Markov networks, and the locality of entanglement spectra
\citep{PoulinHastingsMarkov,KimConditionalIndependence,KatoBrandao,
KuwaharaKatoBrandao,KatoBrandaoEdge}.  Recent work has further clarified the
decay of conditional mutual information in thermal and many-body settings
\citep{kuwahara2024clusteringconditional,kato2025clusteringconditionalmutual}.
However, conditional mutual information is not controlled by ordinary
correlation decay in a purely formal way: conditioning can remove correlations,
but it can also create them.  This makes the question of conditional
independence across a spatial buffer genuinely different from the decay of
ordinary two-point functions or even from the decay of mutual information.

Fermionic Gaussian states again provide a rigid setting in which this question
can be attacked directly.  Entropies of reduced Gaussian states are trace
functions of compressed one-particle Green operators.  This permits an
expansion of the conditional mutual information in terms of powers of the
compressed Green operators.  If the state has exponentially clustering mutual
information, then for fixed finite regions this is equivalent, up to constants
and decay rates, to exponential decay of the corresponding off-diagonal
one-particle blocks; this is the Gaussian analogue of the connection between
information clustering and correlation decay studied in other many-body
settings
\citep{wolf2007arealawsquantumsystems,bluhm2021exponentialdecaymutualinformation}.
Using the trace expansion, we show that the Green-function paths connecting
two outer regions through a growing buffer give a summable contribution.  As a
result, for disjoint translates \(A\) and \(C\) of fixed finite reference
shapes \(A_0\) and \(C_0\), separated by a buffer \(B_r\), we prove
\begin{equation}
        I_\rho(A:C\mid B_r)
        \le
        \frac{K_{A_0,C_0,\rho}}{r},
\end{equation}
where the constant is independent of the buffer scale \(r\).  The
Fawzi--Renner theorem then gives
\begin{equation}
        \inf_{\mathcal R_{B_r\to B_r C}}
        \left\|
        \rho_{A,B_r,C}
        -
        (\id_A\otimes\mathcal R_{B_r\to B_r ,C})(\rho_{A,B_r})
        \right\|_1
        =
        O(r^{-1/2}).
\end{equation}
Thus spatial clustering at the one-particle level yields an approximate-recovery
description for fermionic Gaussian states.

The three approximation results are logically independent, but they share a
common mechanism.  High degree spreads a finite covariance budget over many
interaction edges.  Exchangeability forces all copy-copy covariance blocks to
be identical, and admissibility then makes each of them small.  Spatial
clustering suppresses the covariance paths that contribute to conditional
mutual information across a buffer.  In all three cases, a many-body
approximation statement is obtained by first proving a one-particle covariance
estimate and then converting it into the relevant variational or
information-theoretic quantity.

The remainder of this work is organized as follows.  Section~\ref{sec:preliminaries} fixes the
fermionic covariance-matrix notation, including Majorana covariance matrices,
Gaussian states, the Green operator, and the entropy formula.  Section~\ref{sec:highdegree}
proves covariance monogamy and applies it to product-state approximations for
quadratic Hamiltonians on high-degree graphs, both at zero and positive
temperature.  Section~\ref{sec:definetti} proves the finite fermionic Gaussian
de Finetti theorem and the infinite-exchangeability corollary.  Section~\ref{sec:markov}
proves the conditional mutual information bound for spatially clustering
fermionic Gaussian states and derives the corresponding recoverability
estimate.  Section~\ref{sec:discussion} discusses possible refinements and
extensions.

\section{Preliminaries}\label{sec:preliminaries}

We collect the notation and basic facts about fermionic Gaussian states that
will be used throughout this work.  All systems considered here contain finitely
many fermionic modes unless explicitly stated otherwise.  For an operator
\(X\), we write
\begin{equation}
        \|X\|_p
        =
        \bigl(\tr (X^\dagger X)^{p/2}\bigr)^{1/p},
        \qquad 1\le p<\infty ,
\end{equation}
for the Schatten \(p\)-norm, and \(\|X\|_\infty\) for the operator norm.  For
Hermitian matrices \(A,B\), the notation \(A\succeq0\) means that \(A\) is
positive semidefinite, and \(A\preceq B\) means \(B-A\succeq0\).

\subsection{Fermionic modes, parity, and local regions}

We begin by fixing the operator-algebraic conventions that determine
locality, parity, and signs throughout the covariance-matrix analysis.
An \(m\)-mode fermionic system is described by creation and annihilation
operators
$
        c_1^\dagger,\ldots,c_m^\dagger$ and
        $c_1,\ldots,c_m $,
satisfying the canonical anticommutation relations
\begin{equation}\label{eq:CAR}
        \{c_j,c_k^\dagger\}=\delta_{j,k}\one,
        \qquad
        \{c_j,c_k\}=\{c_j^\dagger,c_k^\dagger\}=0 
\end{equation}
for $j,k=1,\dots, m$.
The Hilbert space is the fermionic Fock space
\(\mathcal F_m\cong(\mathbb C^2)^{\otimes m}\), after choosing a fixed global
ordering of the modes.  This ordering is part of the convention: fermionic
operators supported on disjoint mode sets anticommute before passing to the
even local algebras, and therefore signs in covariance blocks depend on the
chosen global order.
The number operator and parity operator are given by
\begin{equation}
        \widehat N=\sum_{j=1}^m c_j^\dagger c_j,
        \qquad
        P=(-1)^{\widehat N}.
\end{equation}
Physical fermionic states obey the parity superselection rule
\begin{equation}\label{eq:parity-superselection}
        [\rho,P]=0 .
\end{equation}
Equivalently, a physical state is block diagonal with respect to the even and
odd parity sectors.  In particular, expectations of odd monomials in the
fermionic field operators vanish.  Throughout this work, 
all Gaussian states
are understood to be even, hence physical in this sense. 
Number-conserving fermionic states additionally commute with the total particle-number operator.

If \(A\) is a finite set of modes, we denote by \(m_A\) the number of
fermionic modes in \(A\).  The associated real Majorana one-particle space has
dimension
\begin{equation}
        N_A:=2m_A .
\end{equation}
We reserve the symbol \(M_A\) for the covariance matrix of the subsystem
\(A\), and use \(N_A\) for the Majorana one-particle dimension.  The projection
onto the Majorana one-particle space of \(A\) is denoted by \(\Pi_A\).

\subsection{Majorana operators and covariance matrices}

We next express the fermionic operators through Majorana modes and fix the covariance sign convention used in the proofs.
From the creation and annihilation operators we define Hermitian Majorana
operators by
\begin{equation}\label{eq:majorana-definition}
        \gamma_{2j-1}=c_j+c_j^\dagger,
        \qquad
        \gamma_{2j}=i(c_j-c_j^\dagger),
        \qquad j=1,\ldots,m .
\end{equation}
They satisfy the Clifford relations
\begin{equation}\label{eq:majorana-CAR}
        \{\gamma_p,\gamma_q\}=2\delta_{p,q}\one,
        \qquad
        \gamma_p^\dagger=\gamma_p .
\end{equation}
With the convention \eqref{eq:majorana-definition}, the parity operator may be
written as
\begin{equation}\label{eq:parity-majorana}
        P
        =
        (-i)^m\prod_{j=1}^m \gamma_{2j-1}\gamma_{2j}.
\end{equation}

For a physical state \(\rho\), its Majorana covariance matrix, abbreviated
CM, is the real antisymmetric matrix \(M\in\mathbb R^{2m\times2m}\) defined by
\begin{equation}\label{eq:majorana-covariance}
        M_{p,q}
        :=
        \frac{i}{2}\tr\!\left(\rho[\gamma_q,\gamma_p]\right)
        =
        \begin{cases}      \tr\!\left(i\gamma_q\gamma_p\rho\right), & p\ne q,\\[1mm]
        0, & p=q .
        \end{cases}
\end{equation}
This is the sign convention used in Ref.~\citep{NegariSalekExtendibility}.  Notice
in particular that, for \(p\ne q\),
\begin{equation}\label{eq:two-point-from-M}
        \tr(\rho\,\gamma_p\gamma_q)
        =
        iM_{p,q}.
\end{equation}
A matrix \(M\) is a valid fermionic covariance matrix if and only if it is real
antisymmetric and satisfies the bona fide condition
\begin{equation}\label{eq:covariance-positivity}
        \one+iM\succeq0 .
\end{equation}
Equivalently, the eigenvalues of the Hermitian matrix \(iM\) lie in
\([-1,1]\), or, equivalently, all singular values of \(M\) are at most one.
Thus there exists an orthogonal matrix \(O\in SO(2m)\) and real numbers
\(\lambda_1,\ldots,\lambda_m\in[-1,1]\) such that
\begin{equation}\label{eq:canonical-form-M}
        M
        =
        O
        \left(
        \bigoplus_{j=1}^m
        \begin{pmatrix}
        0 & \lambda_j\\
        -\lambda_j & 0
        \end{pmatrix}
        \right)
        O^T .
\end{equation}
For a pure Gaussian state one has \(M^2=-\one\), equivalently
\(|\lambda_j|=1\) for every \(j\).
For a general physical state, the covariance matrix records only the second
moments.  For a fermionic Gaussian state, however, the covariance matrix
determines the state completely.

\subsection{Fermionic Gaussian states}

We now turn to our main object of study, multi-mode Gaussian states.
A parity-preserving state \(\rho\) is called fermionic Gaussian if it can be
written, up to normalization, as the exponential of a quadratic form in
Majorana operators,
\begin{equation}\label{eq:gaussian-exponential}
        \rho
        =
        Z^{-1}
        \exp\!\left(
        \frac{i}{2}\sum_{p,q} h_{p,q}\gamma_p\gamma_q
        \right),
        \qquad
        h=-h^T\in\mathbb R^{2m\times2m},
\end{equation}
with the usual limiting interpretation for pure states.  Equivalently, a state
is Gaussian if and only if all of its higher-order correlations are determined
by its two-point correlations through Wick's rule.  Odd Majorana moments vanish,
and every even Majorana moment is given by the Pfaffian of the corresponding
two-point matrix.

Conversely, every real antisymmetric matrix \(M\) satisfying
\eqref{eq:covariance-positivity} is the covariance matrix of a unique even
fermionic Gaussian state.  We denote this state by \(\rho_M\).  Thus, within
the Gaussian class,
\begin{equation}
        M
        \longleftrightarrow
        \rho_M,
\end{equation}
is a one-to-one parametrization.
We shall use the following covariance-to-state continuity estimate for
fermionic Gaussian states~\citep{BittelMeleEisertLeoneTraceDistance}.  If \(M\) and \(M'\) 
are two valid covariance
matrices, then
\begin{equation}\label{eq:gaussian-covariance-lipschitz-prelim}
        \|\rho_M-\rho_{M'}\|_1
        \le
        \frac12\|M-M'\|_1 .
\end{equation}
This bound allows covariance estimates to be converted directly into
trace-norm estimates for Gaussian states.

\subsection{Restrictions, block structure, and Gaussian product states}

We describe how covariance matrices restrict to regions and how their off-diagonal blocks quantify the failure of Gaussian states to factorize.
The covariance formalism is compatible with restriction to subsystems.  
If
\(A\) is a finite region, then the reduced state \(\rho_A\) of a Gaussian state
\(\rho_M\) is Gaussian with covariance matrix
\begin{equation}\label{eq:restricted-covariance}
        M_A
        =
        \Pi_A M\Pi_A .
\end{equation}
More generally, for two disjoint regions \(A\) and \(B\), written in the
global ordering where the modes of \(A\) precede the modes of \(B\), the
covariance matrix has the block form
\begin{equation}\label{eq:block-covariance}
        M_{A,B}
        =
        \begin{pmatrix}
        M_A & X_{A,B}\\
        -X_{A,B}^T & M_B
        \end{pmatrix}.
\end{equation}
Here
\begin{equation}\label{eq:cross-covariance-block}
        X_{A,B}:=\Pi_A M\Pi_B,
\end{equation}
is the covariance block across the cut \(A|B\).  With the convention
\eqref{eq:majorana-covariance}, if \(a\) labels a Majorana mode in \(A\) and
\(b\) labels a Majorana mode in \(B\), then
\begin{equation}
        (X_{A,B})_{a,b}
        =
        \tr\!\left(i\gamma_b\gamma_a\rho\right).
\end{equation}
For disjoint regions \(A,B_1,\ldots,B_k\), we shall often write
\begin{equation}
        X_i:=
        X_{A, B_i}.
\end{equation}

The product of the Gaussian marginals on \(A\) and \(B\) is the Gaussian state
with covariance matrix
\begin{equation}\label{eq:product-covariance}
        M_A\oplus M_B .
\end{equation}
Thus, in the Gaussian setting, passing from a joint state to the product of
its Gaussian marginals is exactly the operation of deleting the off-diagonal
covariance block \(X_{A,B}\).  In particular,
\begin{equation}
        X_{A,B}=0
        \quad\Longleftrightarrow\quad
        \rho_{M_{A,B}}=\rho_{M_A}\otimes\rho_{M_B},
\end{equation}
for Gaussian states, with the tensor product understood with respect to the
fixed global mode ordering.

The sign convention in \eqref{eq:majorana-covariance} is relevant when writing
expectations of oriented quadratic Hamiltonian terms.  If
\begin{equation}
        H_{A,B}
        =
        i\bm\gamma_A^T K_{A,B}\bm\gamma_B,
        \qquad
        K_{A,B}\in\mathbb R^{N_A\times N_B},
\end{equation}
then
\begin{equation}\label{eq:quadratic-expectation-M}
        \tr(\rho H_{A,B})
        =
        -\tr(K_{A,B}^T X_{A,B}).
\end{equation}
The sign can of course be absorbed into the definition of the oriented
coupling matrix \(K_{A,B}\), and the estimates in this work use only the
absolute value of this pairing.

\subsection{Majorana Green operator and entropy}

The final preliminary connects the covariance matrix to entropy, preparing the conditional mutual-information argument below.
For entropy estimates it is convenient to use the one-particle operator
associated with the covariance matrix.  Given a valid covariance matrix \(M\),
define
\begin{equation}\label{eq:green-definition}
        G
        :=
        \frac{\one+iM}{2}.
\end{equation}
Then
\begin{equation}
        0\preceq G\preceq\one,
        \qquad
        G^\dagger=G,
        \qquad
        M=-i(2G-\one).
\end{equation}
Since \(M\) is real and antisymmetric, \(G\) also satisfies the particle-hole
symmetry
\begin{equation}\label{eq:particle-hole-G}
        G^T=\one-G .
\end{equation}
For a finite region \(A\), we write
\begin{equation}\label{eq:restricted-G}
        G_A
        :=
        \Pi_A G\Pi_A .
\end{equation}
If \(A\) contains \(m_A\) fermionic modes, then
\begin{equation}
        \rank\Pi_A=N_A=2m_A,
        \qquad
        \tr G_A=m_A .
\end{equation}
The entropy of the reduced Gaussian state \(\rho_A\) is
\begin{equation}\label{eq:entropy-green}
        S(\rho_A)
        =
        -\tr\bigl(G_A\log G_A\bigr),
\end{equation}
where the trace is taken over the Majorana one-particle space of \(A\).  The
particle-hole symmetry \eqref{eq:particle-hole-G} ensures that this expression
counts each physical fermionic mode correctly: the eigenvalues of \(G_A\) come
in pairs \(\lambda,1-\lambda\).

For two finite regions \(A\) and \(B\), we also use the block notation
\begin{equation}
        G_{A,B}:=\Pi_A G\Pi_B .
\end{equation}
In spatial sections, when \(x,y\) are lattice sites or finite cells,
\(G_{x, y}\) denotes the corresponding one-particle block from the Majorana
space over \(y\) to the Majorana space over \(x\).  Thus \(G_X\) is the
compression of \(G\) to \(X\), whereas \(G_{x, y}\) is the off-diagonal block
between the sites \(x\) and \(y\).

\section{Product-state approximation on high-degree graphs}
\label{sec:highdegree}

A central problem in many-body theory is to identify variational classes that
approximate ground-state energies without requiring a description of the full
many-body state.  Product states are the simplest possible ansatz.  For a
general local Hamiltonian, however, there is no reason for such an ansatz to be
accurate: local interactions may create entanglement on many edges, and the
ground-state energy can depend essentially on these correlations.  It is
therefore important to understand regimes in which the variational problem
simplifies and product states become approximately optimal.

One such regime is the dense, or high-degree, interaction-graph limit.  If each
local system interacts with many others, then in the normalized Hamiltonian the
contribution of any single edge is small in the global energy.  Brand\~ao and
Harrow \citep{BrandaoHarrowGround,BrandaoHarrowStates} have made this principle
quantitative for general qudit Hamiltonians: for local dimension
\(d_{\rm loc}\), the product-state approximation error scales as
\begin{equation}
        \left(\frac{d_{\rm loc}^2\log d_{\rm loc}}{D}\right)^{1/3},
\end{equation}
where \(D\) is the degree of the interaction graph.

The purpose of this section is to prove a substantially sharper version of this phenomenon
for quadratic fermionic Hamiltonians.  The improvement comes from the
covariance-matrix structure 
of fermionic systems.  For a quadratic fermionic
interaction, the energy of an edge is determined entirely by the two-point
Majorana covariance block across that edge. We show that a fixed fermionic
region can share only a bounded total amount of two-point covariance with many
disjoint regions.  Thus, on a high-degree graph, the average covariance across
an interaction edge is small.  This gives a product-state approximation with
error scaling as
\begin{equation}
        \left(\frac{\log d_{\rm loc}}{D}\right)^{1/2}.
\end{equation}
In particular, the dependence on \(d_{\rm loc}\) improves from polynomial to
polylogarithmic, while the dependence on the degree improves from \(D^{-1/3}\)
to \(D^{-1/2}\).

Conceptually, this is the exactly quadratic counterpart of the interacting
fermionic approximation theorem in our earlier work
\citep{KrumnowZimborasEisertGroundApprox}: the same obstruction to distributing
fermionic correlations is used here at the covariance level, where it yields a
stronger rate and a substantially shorter argument.
We use the high-degree normalization of Brand\~ao and Harrow
\citep{BrandaoHarrowGround,BrandaoHarrowStates}.  Let
\(\mathcal G=(V,E)\) be a finite simple \(D\)-regular graph, with \(D\ge1\),
on \(n=|V|\) vertices.  In particular, the graph has neither loops nor
parallel edges.  The
vertices label local systems \(Q_1,\ldots,Q_n\), each of local dimension
\(d_{\rm loc}\), and the edges specify the two-body interactions.  We write the
Hamiltonian as the edge average
\begin{equation}
        H
        =
        \mathbb E_{(i,j)\in E} H_{i,j}
        =
        \frac{2}{nD}\sum_{(i,j)\in E} H_{i,j},
        \label{eq:BH-normalized-H}
\end{equation}
where \(H_{i,j}\) is supported on \(Q_i\cup Q_j\), and
\(\|H_{i,j}\|\le 1\) for every edge.  With this convention,
\(\tr(H\rho)\) is an intensive energy and the local interaction terms obey a
degree-independent operator-norm bound.

We now restrict to the quadratic fermionic case.  Each \(Q_i\) is the Fock
space of \(m\) fermionic modes, so that
$ d_{\rm loc}=2^m $.
Let
\begin{equation}
        \bm\gamma_i
        =
        (\gamma_{i,1},\ldots,\gamma_{i,2m})^T,
\end{equation}
denote the vector of Majorana operators at vertex \(i\).  After fixing an
arbitrary orientation of each edge, we write the inter-site interaction as
\begin{equation}
        H_{i,j}
        =
        i\bm\gamma_i^T K_{i,j}\bm\gamma_j,
        \qquad
        K_{i,j}\in\mathbb R^{2m\times 2m}.
        \label{eq:fermionic-edge-Hij}
\end{equation}
Here the factor of \(i\) makes the operator self-adjoint, since Majorana
operators on distinct sites anticommute.  The chosen orientation is only a
bookkeeping convention for writing \(K_{i,j}\); the estimates below depend only
on norms and are independent of this choice.

To proceed, let
\begin{equation}
        e_0(H)
        :=
        \min_{\rho}\tr(H\rho)
        \label{eq:ground-energy-definition},
\end{equation}
be the ground-state energy.  The high-degree product approximation for
quadratic fermions is the following.
Here and below, variational optimizations are over physical parity-preserving
fermionic states.

\begin{theorem}[High-degree product approximation for quadratic fermions]
\label{thm:fermionic-product-state-approximation}
Let \(\mathcal G=(V,E)\) be a finite simple \(D\)-regular graph, with
\(D\ge1\), on \(n\) vertices, and let
\(H\) be a quadratic fermionic Hamiltonian of the form
\eqref{eq:BH-normalized-H}--\eqref{eq:fermionic-edge-Hij}.  Assume that
\(\|H_{i,j}\|\le 1\) for every \((i,j)\in E\).  Then there exists a physical
product state
\begin{equation}
        \varphi
        =
        \varphi^{Q_1}\otimes\cdots\otimes\varphi^{Q_n},
\end{equation}
such that
\begin{equation}
        \tr(H\varphi)
        \le
        e_0(H)
        +
        \sqrt{\frac{2m}{D}}
        =
        e_0(H)
        +
        \sqrt{\frac{2\log_2 d_{\rm loc}}{D}} .
        \label{eq:fermionic-product-state-bound}
\end{equation}
In fact, each factor of \(\varphi\) may be chosen as a pure parity eigenstate.
\end{theorem}

The proof isolates the genuinely fermionic input in a covariance monogamy
bound.  This bound is not a property of the Hamiltonian and does not require
Gaussianity; it follows from the admissibility condition for Majorana
covariance matrices.  Once this monogamy estimate is established, the proof of
Theorem~\ref{thm:fermionic-product-state-approximation} is direct: applying
the bound at each vertex controls the average covariance across an interaction
edge, and the quadratic form \eqref{eq:fermionic-edge-Hij} converts this
covariance estimate into an energy estimate.

\subsection{Covariance monogamy}
\label{subsec:covariance-monogamy}

The first ingredient is a monogamy bound for fermionic two-point correlations.
If one region \(A\) is correlated with many disjoint regions
\(B_1,\ldots,B_k\), then the covariance blocks from \(A\) to these regions
cannot all be large.  More precisely, the total covariance budget of \(A\) is
finite, and therefore the average two-point covariance between \(A\) and one of
the regions \(B_i\) goes to zero as the number of regions grows.

Let \(\rho\) be a physical fermionic state on \(A, B_1,\cdots ,B_k\), not
necessarily Gaussian. We can write its  
Majorana covariance matrix in block form as
\begin{equation}
M_{AB_1\cdots B_k}
=
\begin{pmatrix}
M_A & X_1 & X_2 & \cdots & X_k\\
-X_1^T & M_{B_1} & * & \cdots & *\\
-X_2^T & * & M_{B_2} & \cdots & *\\
\vdots & \vdots & \vdots & \ddots & \vdots\\
-X_k^T & * & * & \cdots & M_{B_k}
\end{pmatrix}.
\label{eq:monogamy-block-form}
\end{equation}
Here
\begin{equation}
        X_i=\Pi_A M\Pi_{B_i},
\end{equation}
is the two-point covariance block between \(A\) and \(B_i\), and the starred
blocks denote covariance blocks among the regions \(B_1,\ldots,B_k\), which
will not be used.

\begin{theorem}[Sharp monogamy of fermionic two-point correlations]
\label{thm:monogamy}
Let \(A\) contain \(m_A\) fermionic modes.  Let \(\rho\) be a physical
fermionic state, not necessarily Gaussian, on disjoint regions
\(A,B_1,\dots,B_k\), and let \(M\) be its Majorana covariance matrix.  Then
\begin{equation}\label{eq:opineq}
\sum_{i=1}^{k}X_iX_i^{T}\;\preceq\;\one_A ,
\end{equation}
and consequently
\begin{equation}\label{eq:trineq}
\E_i\,\HS{X_i}^{2}
=\frac{1}{k}\sum_{i=1}^{k}\HS{X_i}^{2}
\;\le\;\frac{2m_A}{k},
\end{equation}
where $i$ is chosen uniformly from $\{1,\dots,k\}$. Both bounds are optimal,
see Example~\ref{ex:sharp}.
\end{theorem}

\begin{proof}
The projections $\Pi_{B_1},\dots,\Pi_{B_k}$ have mutually orthogonal ranges,
so $\Pi_B:=\sum_{i=1}^{k}\Pi_{B_i}$ is the orthogonal projection onto the
Majorana one-particle space of $B:=B_1\cup\dots\cup B_k$, and
$\Pi_B\preceq\one$. Using $\Pi_{B_i}\Pi_{B_i}^{T}=\Pi_{B_i}$,
\begin{equation}
\sum_{i=1}^{k}X_iX_i^{T}
=\sum_{i=1}^{k}\Pi_A M\Pi_{B_i}M^{T}\Pi_A
=\Pi_A M\,\Pi_B\,M^{T}\Pi_A
\;\preceq\;\Pi_A MM^{T}\Pi_A
\;\preceq\;\Pi_A ,
\end{equation}
where the first inequality is the congruence of $\Pi_B\preceq\one$ by $M$, and the second follows from $MM^{T}\preceq\one$. Restricting to the range of $\Pi_A$ yields \eqref{eq:opineq}. Taking the trace gives
\begin{equation}
\sum_{i=1}^{k}\HS{X_i}^{2}
=\tr\Bigl(\sum_{i=1}^{k}X_iX_i^{T}\Bigr)
\le\tr\Pi_A=N_A=2m_A ,
\end{equation}
and dividing by $k$ proves \eqref{eq:trineq}.
\end{proof}

\begin{example}[Optimality for every $k$]\label{ex:sharp}
Let each $B_i$ contain $m_A$ modes, and consider the real antisymmetric
matrix with vanishing diagonal and $B_iB_j$ blocks,
\begin{equation}
M=\begin{pmatrix}
0 & X_1 & \cdots & X_k\\
-X_1^{T} & 0 & \cdots & 0\\
\vdots & \vdots & \ddots & \vdots\\
-X_k^{T} & 0 & \cdots & 0
\end{pmatrix},
\qquad
X_i=k^{-1/2}\,\one_{N_A}.
\end{equation}
Then $MM^{T}=\one_{N_A}\oplus\bigl(\tfrac{1}{k}J_k\otimes\one_{N_A}\bigr)$,
where $J_k$ is the $k\times k$ all-ones matrix. Since $\tfrac1k J_k$ has
eigenvalues $1$ and $0$, we get
$\mathrm{spec}(MM^{T})\subset\{0,1\}$, so $\|M\|_{\infty}=1$ and $M$ is the covariance matrix of a unique fermionic Gaussian state. For this state,
\begin{equation}
\sum_{i=1}^{k}X_iX_i^{T}=\one_A,
\qquad
\HS{X_i}^{2}=\frac{N_A}{k}=\frac{2m_A}{k}
\quad\text{for every }i,
\end{equation}
so \eqref{eq:opineq} is saturated and \eqref{eq:trineq} holds with equality
termwise. Hence neither the operator bound, nor the constant $2m_A$, nor the
$1/k$ scaling can be improved.
\end{example}

\subsection{Proof of the product-state approximation}
\label{subsec:proof-product-approximation}

We now use the covariance monogamy bound to prove
Theorem~\ref{thm:fermionic-product-state-approximation}.  Since \(H\) is a
finite quadratic fermionic Hamiltonian, it has a fermionic Gaussian ground
state.  Choose such a ground state \(\rho_0\), so that
\begin{equation}
        \tr(H\rho_0)=e_0(H).
\end{equation}
For an oriented edge \((i,j)\in E\), let \(M^{(0)}\) be the covariance matrix
of \(\rho_0\), and define
\begin{equation}
        X_{i,j}^{(0)}
        :=
        \Pi_i M^{(0)}\Pi_j .
\end{equation}
Thus, with the convention of \eqref{eq:majorana-covariance},
\begin{equation}
        X_{i,j}^{(0)}(a,b)
        =
        \tr\!\left(i\gamma_{j,b}\gamma_{i,a}\rho_0\right).
\end{equation}
Consequently,
\begin{equation}
        \tr\!\left(\rho_0\,i\gamma_{i,a}\gamma_{j,b}\right)
        =
        -X_{i,j}^{(0)}(a,b).
\end{equation}
Since
\begin{equation}
        H_{i,j}
        =
        i\bm\gamma_i^T K_{i,j}\bm\gamma_j,
\end{equation}
we have
\begin{equation}
        \tr(\rho_0 H_{i,j})
        =
        -\tr\!\left(K_{i,j}^T X_{i,j}^{(0)}\right).
        \label{eq:edge-expectation-covariance}
\end{equation}

We first translate the operator normalization of \(H_{i,j}\) into the
one-particle bound needed for the covariance estimate.  Let \(s_\alpha\) be
the singular values of \(K_{i,j}\).  By applying orthogonal changes of
Majorana basis at the two endpoints, the edge term is brought to the form
\begin{equation}
        \sum_\alpha s_\alpha\, i\eta_\alpha\widetilde\eta_\alpha ,
\end{equation}
where the operators \(i\eta_\alpha\widetilde\eta_\alpha\) are commuting
self-adjoint involutions.  Hence
\begin{equation}
        \|H_{i,j}\|
        =
        \sum_\alpha s_\alpha
        =
        \|K_{i,j}\|_1 .
\end{equation}
In particular,
\begin{equation}
        \|K_{i,j}\|_2
        \le
        \|K_{i,j}\|_1
        =
        \|H_{i,j}\|
        \le
        1 .
\end{equation}
Therefore, by \eqref{eq:edge-expectation-covariance} and Cauchy--Schwarz,
\begin{equation}
        \left|\tr(\rho_0 H_{i,j})\right|
        =
        \left|\tr\!\left(K_{i,j}^T X_{i,j}^{(0)}\right)\right|
        \le
        \|K_{i,j}\|_2\,\|X_{i,j}^{(0)}\|_2
        \le
        \|X_{i,j}^{(0)}\|_2 .
        \label{eq:edge-energy-by-X}
\end{equation}

Now fix a vertex \(i\).  Applying Theorem~\ref{thm:monogamy} with
\(A=Q_i\) and with \(B_j=Q_j\) for \(j\in N(i)\), where \(N(i)\) denotes the
neighbourhood of \(i\), gives
\begin{equation}
        \sum_{j\in N(i)}
        \|\Pi_i M^{(0)}\Pi_j\|_2^2
        \le
        2m .
\end{equation}
Summing this estimate over all \(i\in V\) counts each undirected edge twice.
The two orientations give covariance blocks that differ only by a transpose
and a sign, and therefore have the same Hilbert--Schmidt norm.  Since
\(|E|=nD/2\), we obtain
\begin{equation}
        \mathbb E_{(i,j)\in E}
        \|X_{i,j}^{(0)}\|_2^2
        \le
        \frac{2m}{D}.
\end{equation}
Jensen's inequality then gives
\begin{equation}
        \mathbb E_{(i,j)\in E}
        \|X_{i,j}^{(0)}\|_2
        \le
        \sqrt{\frac{2m}{D}} .
        \label{eq:average-edge-X-bound}
\end{equation}

Let
\begin{equation}
        \sigma
        =
        \rho_0^{Q_1}\otimes\cdots\otimes\rho_0^{Q_n},
\end{equation}
be the product of the one-site marginals of \(\rho_0\).  At the covariance
level, passing from \(\rho_0\) to \(\sigma\) keeps all one-site covariance
blocks and deletes all covariance blocks between distinct sites.  Therefore,
for every edge \((i,j)\in E\), the covariance block of \(\sigma\) between
\(i\) and \(j\) is zero.  Using \eqref{eq:edge-expectation-covariance}, we get
\begin{equation}
        \tr(\sigma H_{i,j})-\tr(\rho_0 H_{i,j})
        =
        \tr\!\left(K_{i,j}^T X_{i,j}^{(0)}\right).
\end{equation}
Averaging over edges and using \eqref{eq:average-edge-X-bound}, we find
\begin{equation}
\begin{aligned}
        \tr(H\sigma)-e_0(H)
        &=
        \mathbb E_{(i,j)\in E}
        \left[
        \tr(\sigma H_{i,j})-\tr(\rho_0 H_{i,j})
        \right]  \\
        &\le
        \mathbb E_{(i,j)\in E}
        \left|
        \tr\!\left(K_{i,j}^T X_{i,j}^{(0)}\right)
        \right|  \\
        &\le
        \mathbb E_{(i,j)\in E}
        \|X_{i,j}^{(0)}\|_2
        \le
        \sqrt{\frac{2m}{D}} .
\end{aligned}
\end{equation}
Thus
\begin{equation}
        \tr(H\sigma)
        \le
        e_0(H)+\sqrt{\frac{2m}{D}} .
\end{equation}

The state \(\sigma\) is a product state, but it need not be pure.  Each
one-site marginal \(\rho_0^{Q_i}\) commutes with the corresponding local
parity operator and therefore admits a spectral decomposition into pure parity
eigenstates.  Consequently, \(\sigma\) is a convex combination of physical
pure product states.  The energy is affine in the state, so at least one such
pure product state
\begin{equation}
        \varphi
        =
        \varphi^{Q_1}\otimes\cdots\otimes\varphi^{Q_n},
\end{equation}
in this decomposition satisfies
\begin{equation}
        \tr(H\varphi)\le \tr(H\sigma).
\end{equation}
Consequently,
\begin{equation}
        \tr(H\varphi)
        \le
        e_0(H)+\sqrt{\frac{2m}{D}}
        =
        e_0(H)+\sqrt{\frac{2\log_2 d_{\rm loc}}{D}} .
\end{equation}
This proves Theorem~\ref{thm:fermionic-product-state-approximation}.

\begin{remark}[On-site terms]
\label{rem:onsite-terms}
The theorem is stated for Hamiltonians whose edge terms are purely inter-site
quadratic fermionic interactions.  This is only for notational simplicity.  The
same conclusion holds if one adds arbitrary even on-site terms,
\begin{equation}
        H_{\rm loc}
        =
        \frac1n\sum_{i\in V} h_i,
\end{equation}
where each \(h_i\) is an even self-adjoint operator on \(Q_i\), equivalently
\([h_i,P_i]=0\), with \(P_i\) the local parity operator, but need not be
quadratic.  Indeed, in the proof above the approximating product state is
obtained by replacing a state \(\rho\) with the product of its
one-site marginals,
\begin{equation}
        \rho^{Q_1}\otimes\cdots\otimes\rho^{Q_n}.
\end{equation}
This replacement preserves every one-site reduced state, and hence preserves
the expectation of \(H_{\rm loc}\).  The only terms whose expectation can
change are the inter-site quadratic terms, and these are exactly the terms
controlled by the covariance monogamy bound.  Thus the hypothesis needed for
the argument is not that the whole Hamiltonian be quadratic, but only that the
inter-site interactions along the graph are quadratic fermionic operators.
For the enlarged Hamiltonian one starts with any physical ground state;
Gaussianity is not needed because Theorem~\ref{thm:monogamy} applies to
arbitrary physical states.
\end{remark}

\subsection{Finite-temperature free energy}
\label{subsec:free-energy-product-approximation}

The covariance-monogamy mechanism is not confined to ground states.  It also
shows that the thermodynamic variational problem becomes effectively
mean-field-like on interaction graphs of high degree.  To make this precise,
let \(\beta>0\) denote the inverse temperature and define the free-energy
functional by
\begin{equation}
        F_\beta(\omega)
        =
        \tr(H\omega)-\beta^{-1}S(\omega),
\end{equation}
where
\begin{equation}
        S(\omega)=-\tr(\omega\log\omega),
\end{equation}
is the von Neumann entropy.  Here \(\beta\) is conjugate to the normalized,
intensive Hamiltonian \(H\).  Equivalently, with \(\widehat H=nH\) and
\(\beta_{\rm phys}=\beta/n\), the functional \(F_\beta\) is the physical free
energy per site for \((\widehat H,\beta_{\rm phys})\).  Let
\begin{equation}
        \rho_\beta
        =
        \frac{e^{-\beta H}}{\tr(e^{-\beta H})},
\end{equation}
be the Gibbs state of \(H\).  Since \(H\) is quadratic, \(\rho_\beta\) is a
fermionic Gaussian state.

\begin{theorem}[Fermionic Gaussian product approximation for free energy]
\label{thm:fermionic-free-energy-product}
Let \(\mathcal G=(V,E)\) be a finite simple \(D\)-regular graph, with
\(D\ge1\), on \(n\) vertices, and let
\(H\) be a quadratic fermionic Hamiltonian of the form
\eqref{eq:BH-normalized-H}--\eqref{eq:fermionic-edge-Hij}.  Assume that
\(\|H_{i,j}\|\le 1\) for every \((i,j)\in E\).  For \(\beta>0\), let
\(\rho_\beta\) be the Gibbs state of \(H\), and define the product of its
one-site marginals by
\begin{equation}
        \sigma_\beta
        =
        \rho_\beta^{Q_1}\otimes\cdots\otimes\rho_\beta^{Q_n}.
\end{equation}
Then \(\sigma_\beta\) is a Gaussian product state and
\begin{equation}
        0
        \le
        F_\beta(\sigma_\beta)-F_\beta(\rho_\beta)
        \le
        \sqrt{\frac{2m}{D}}
        =
        \sqrt{\frac{2\log_2 d_{\rm loc}}{D}} .
        \label{eq:fermionic-free-energy-product-bound}
\end{equation}
\end{theorem}

\begin{proof}
The lower bound follows from the Gibbs variational principle, since
\(\rho_\beta\) minimizes \(F_\beta\).
Thus, the nontrivial step is to prove the upper bound.  Let \(M^{(\beta)}\) be the covariance matrix of
\(\rho_\beta\).  For an oriented edge \((i,j)\in E\), set
\begin{equation}
        X_{i,j}^{(\beta)}
        :=
        \Pi_i M^{(\beta)}\Pi_j .
\end{equation}
As in the proof of
Theorem~\ref{thm:fermionic-product-state-approximation}, the operator
normalization \(\|H_{i,j}\|\le1\) implies
\begin{equation}
        \|K_{i,j}\|_2\le1.
\end{equation}
Using the covariance convention \eqref{eq:majorana-covariance}, we have
\begin{equation}
        \tr(\rho_\beta H_{i,j})
        =
        -\tr\!\left(K_{i,j}^T X_{i,j}^{(\beta)}\right),
\end{equation}
and therefore
\begin{equation}
        \left|
        \tr\!\left(K_{i,j}^T X_{i,j}^{(\beta)}\right)
        \right|
        \le
        \|X_{i,j}^{(\beta)}\|_2 .
\end{equation}
Applying Theorem~\ref{thm:monogamy} at each vertex and averaging over the
edges gives
\begin{equation}
        \mathbb E_{(i,j)\in E}
        \left|
        \tr\!\left(K_{i,j}^T X_{i,j}^{(\beta)}\right)
        \right|
        \le
        \sqrt{\frac{2m}{D}} .
        \label{eq:gibbs-average-edge-bound}
\end{equation}

Now compare \(\rho_\beta\) with
\begin{equation}
        \sigma_\beta
        =
        \rho_\beta^{Q_1}\otimes\cdots\otimes\rho_\beta^{Q_n}.
\end{equation}
Passing from \(\rho_\beta\) to \(\sigma_\beta\) keeps the one-site marginals
and deletes the covariance blocks between distinct sites.  Therefore,
for every edge \((i,j)\in E\),
\begin{equation}
        \tr(\sigma_\beta H_{i,j})-\tr(\rho_\beta H_{i,j})
        =
        \tr\!\left(K_{i,j}^T X_{i,j}^{(\beta)}\right).
\end{equation}
Averaging over edges and using \eqref{eq:gibbs-average-edge-bound}, we obtain
\begin{equation}
        \tr(H\sigma_\beta)-\tr(H\rho_\beta)
        \le
        \sqrt{\frac{2m}{D}} .
\end{equation}

It remains to use the entropy term.  Since \(\sigma_\beta\) is the product of
the one-site marginals of \(\rho_\beta\), subadditivity gives
\begin{equation}
        S(\sigma_\beta)
        =
        \sum_{i=1}^n S(\rho_\beta^{Q_i})
        \ge
        S(\rho_\beta).
\end{equation}
Hence
\begin{equation}
\begin{aligned}
        F_\beta(\sigma_\beta)-F_\beta(\rho_\beta)
        &=
        \tr(H\sigma_\beta)-\tr(H\rho_\beta)
        -
        \beta^{-1}\bigl(S(\sigma_\beta)-S(\rho_\beta)\bigr)  \\
        &\le
        \tr(H\sigma_\beta)-\tr(H\rho_\beta)  \\
        &\le
        \sqrt{\frac{2m}{D}} .
\end{aligned}
\end{equation}
This proves Eq.~\eqref{eq:fermionic-free-energy-product-bound}.
\end{proof}

\section{Exchangeability and fermionic Gaussian de Finetti bounds}
\label{sec:definetti}

We now turn to the role of exchangeability.  In ordinary quantum de Finetti
theory, permutation symmetry, together with extendibility to a sufficiently
large number of copies, leads to approximate product structure.  In the
fermionic Gaussian setting, this principle admits a particularly direct and
rigid formulation at the level of covariance matrices.  Graded permutation
symmetry forces all off-diagonal covariance blocks between distinct copies to
coincide, while the bona fide condition for the full \(n\)-copy covariance
matrix suppresses their magnitude to order \(1/n\).

The proof is entirely one-particle.  We first state the finite de Finetti
theorem.  We then write the exchangeable covariance matrix in block form and
diagonalize its copy-space part into the symmetric direction and its orthogonal
complement.  The two resulting one-copy covariance matrices differ by exactly
\(nY\), where \(Y\) is the common off-diagonal covariance block.  This gives the
covariance suppression bound, which is then converted into a trace-norm bound
between Gaussian states.

Let \(A\) be a fermionic system with \(m_A\) modes, equivalently \(N_A=2m_A\)
Majorana operators.  We write \(A_1,\ldots,A_k\) for \(k\) identical copies of
\(A\), with the global ordering of copies fixed once and for all.
For \(\pi\in S_n\), let \(U_\pi\) be the graded fermionic copy-permutation
unitary determined by
\begin{equation}
        U_\pi^\dagger\gamma_a^{(i)}U_\pi
        =
        \gamma_a^{(\pi(i))}.
        \label{eq:graded-copy-permutation}
\end{equation}
All permutation invariance in this section refers to these fermionic
permutation unitaries, rather than to ordinary tensor-factor permutations
after a Jordan--Wigner identification.

\begin{definition}[\(n\)-exchangeability]
Let \(1\le k\le n\).  A state \(\rho_{A_1, \cdots , A_k}\) is called
\(n\)-exchangeable if there exists a physical state
\(\widetilde\rho_{A_1\cdots A_n}\) which satisfies
\(U_\pi\widetilde\rho U_\pi^\dagger=\widetilde\rho\) for every \(\pi\in S_n\)
and whose marginal on \(A_1\cdots A_k\) is
\(\rho_{A_1, \cdots , A_k}\).  A state is infinitely exchangeable if it is
\(n\)-exchangeable for every \(n\ge k\).
\end{definition}

The main result of this section is the following finite de Finetti bound.

\begin{theorem}[Finite fermionic Gaussian de Finetti bound]
\label{thm:exchangeable-gaussian-definetti}
Let \(\rho_{A_1, \cdots , A_k}\) be an \(n\)-exchangeable fermionic Gaussian state,
and let \(\rho_A\) be its one-copy marginal.  Then there exists a probability
measure \(\mu\) on the set of physical states of \(A\) such that
\begin{equation}
        \trnorm{
        \rho_{A_1, \cdots , A_k}
        -
        \int \omega^{\otimes k}\,d\mu(\omega)
        }
        \le
        \frac{4m_A(k-1)}{n}.
        \label{eq:finite-definetti-bound}
\end{equation}
In fact, one may take \(\mu=\delta_{\rho_A}\), so that the approximating
de Finetti state is simply \(\rho_A^{\otimes k}\).
\end{theorem}

The case \(k=1\) is immediate.  Hence, in the covariance analysis below, we
assume \(2\le k\le n\).
The remainder
of this section proves
Theorem~\ref{thm:exchangeable-gaussian-definetti}.  For fermionic Gaussian
states, exchangeability can be checked directly at the covariance level.
Given any permutation-invariant extension \(\widetilde\rho\), let
\(\widetilde M\) be its covariance matrix and let \(\rho_{\widetilde M}\) be
the unique Gaussian state with this covariance matrix.  The matrix
\(\widetilde M\) is permutation invariant, and uniqueness implies that
\(\rho_{\widetilde M}\) is permutation invariant as well.  Its \(k\)-copy
marginal has the covariance matrix of \(\rho_{A_1, \cdots , A_k}\); since the
prescribed marginal is Gaussian, uniqueness shows that this marginal of
\(\rho_{\widetilde M}\) is precisely \(\rho_{A_1, \cdots, A_k}\).  Thus an
arbitrary permutation-invariant extension may be replaced by a
permutation-invariant Gaussian extension with the same marginal.  Conversely, every admissible
permutation-invariant covariance extension determines such a Gaussian state
extension.

\subsection{Permutation-invariant covariance matrices}
\label{subsec:exchangeable-sector-suppression}

We now translate graded permutation symmetry into the rigid block
structure that drives the quantitative de Finetti bound.
Let \(\rho_{A_1, \cdots , A_k}\) be permutation-invariant and Gaussian, and let
\(M^{(k)}\) be its Majorana covariance matrix.  Since all copies are identical,
permutation invariance implies that all diagonal blocks are equal to the
one-copy covariance matrix \(M_A\), and that all off-diagonal copy-copy blocks
are described by a single matrix \(Y\).

There is one fermionic point to note.  A priori, the two-copy covariance matrix
has the block form
\begin{equation}
        \begin{pmatrix}
        M_A & Y\\
        -Y^T & M_A
        \end{pmatrix}.
\end{equation}
However, for a permutation-invariant fermionic extension, the off-diagonal
block \(Y\) is itself antisymmetric.  Indeed, with the covariance convention
\begin{equation}
        M_{p,q}
        =
        \tr\!\left(i\gamma_q\gamma_p\rho\right)
        \qquad (p\ne q),
\end{equation}
we have
\begin{equation}
        Y_{A,B}
        =
        \tr\!\left(
        i\gamma^{(2)}_b\gamma^{(1)}_a\rho
        \right).
\end{equation}
Exchanging the two copies gives
\begin{equation}
        Y_{A,B}
        =
        \tr\!\left(
        i\gamma^{(1)}_b\gamma^{(2)}_a\rho
        \right)
        =
        -
        \tr\!\left(
        i\gamma^{(2)}_a\gamma^{(1)}_b\rho
        \right)
        =
        -Y_{b,a},
\end{equation}
where the minus sign comes from the canonical anticommutation relations for
Majorana operators on distinct fermionic copies.  Hence \(Y^T=-Y\).  This is
the covariance-matrix version of the antisymmetry condition for fermionic
symmetric extensions; see, for example, the covariance characterization of
fermionic extendibility in \citep{NegariSalekExtendibility}.  Consequently
\(-Y^T=Y\), and the permutation-invariant \(k\)-copy covariance matrix may be
written as
\begin{equation}
        M^{(k)}
        =
        \begin{pmatrix}
        M_A & Y & \cdots & Y\\
        Y & M_A & \cdots & Y\\
        \vdots & \vdots & \ddots & \vdots\\
        Y & Y & \cdots & M_A
        \end{pmatrix}.
        \label{eq:exchangeable-block-matrix}
\end{equation}
Here both \(M_A\) and \(Y\) are real antisymmetric matrices.

The proof therefore reduces to showing that \(n\)-exchangeability forces
\(\trnorm{Y}=O(m_A/n)\).  Once this is known, the covariance matrix
\eqref{eq:exchangeable-block-matrix} is close to the product covariance
\(M_A^{\oplus k}\), and the covariance-to-state Lipschitz bound for fermionic
Gaussian states gives \eqref{eq:finite-definetti-bound}.

Let \(M^{(n)}\) be a permutation-invariant \(n\)-copy covariance matrix
extending \(M^{(k)}\).  Equivalently,
\begin{equation}
        M^{(n)}
        =
        \one_n\otimes M_A
        +
        (J_n-\one_n)\otimes Y,
\end{equation}
where \(J_n\) is the \(n\times n\) all-ones matrix and \(\one_n\) is the
\(n\times n\) identity.  Define
\begin{equation}
        \ket{+}
        :=
        \frac1{\sqrt n}(1,\ldots,1)^T,
        \qquad
        P_+:=\ket{+}\!\bra{+},
        \qquad
        P_-:=\one_n-P_+ .
\end{equation}
The vector \(\ket{+}\) spans the symmetric copy-space sector.  Since
\(J_n-\one_n\) has eigenvalue \(n-1\) on this sector and eigenvalue \(-1\) on
its orthogonal complement, the covariance matrix decomposes as
\begin{equation}
        M^{(n)}
        =
        P_+\otimes\bigl(M_A+(n-1)Y\bigr)
        +
        P_-\otimes(M_A-Y).
        \label{eq:exchangeable-sector-decomposition}
\end{equation}
Thus the \(n\)-copy bona fide condition reduces to two one-copy bona fide
conditions.

\begin{proposition}[Covariance criterion for exchangeability]
\label{prop:exchangeability-criterion}
Let \(n\ge2\).  The block-constant covariance matrix \(M^{(n)}\) is admissible
if and only if
\begin{equation}
        \one+i\bigl(M_A+(n-1)Y\bigr)\succeq0,
        \qquad
        \one+i(M_A-Y)\succeq0 .
        \label{eq:exchangeability-criterion}
\end{equation}
Equivalently, the two real antisymmetric matrices
\(M_A+(n-1)Y\) and \(M_A-Y\) are bona fide fermionic covariance matrices.
\end{proposition}

\begin{proof}
A real antisymmetric matrix \(M\) is a valid Majorana covariance matrix exactly
when \(\one+iM\succeq0\).  Applying this condition to \(M^{(n)}\) and using
\eqref{eq:exchangeable-sector-decomposition}, we obtain
\begin{equation}
\begin{aligned}
        \one+iM^{(n)}
        &=
        P_+\otimes
        \Bigl(\one+i\bigl(M_A+(n-1)Y\bigr)\Bigr)  \\
        &\quad+
        P_-\otimes
        \Bigl(\one+i(M_A-Y)\Bigr).
\end{aligned}
\end{equation}
Since \(P_+\) and \(P_-\) have orthogonal ranges, this matrix is positive if
and only if its restrictions to the two sectors are positive.  This is exactly
\eqref{eq:exchangeability-criterion}.
\end{proof}

We now extract the quantitative consequence.  This is the step where
exchangeability produces the small parameter \(1/n\).

\begin{lemma}[Suppression of the off-diagonal covariance block]
\label{lem:exchangeable-block-bound}
Let \(2\le k\le n\).  If \(M^{(k)}\) is the \(k\)-copy marginal of an admissible
permutation-invariant covariance matrix on \(n\) copies, then
\begin{equation}
        \trnorm{Y}
        \le
        \frac{4m_A}{n}.
        \label{eq:Y-trace-bound}
\end{equation}
\end{lemma}

\begin{proof}
By Proposition~\ref{prop:exchangeability-criterion}, the matrices
\(M_A+(n-1)Y\) and \(M_A-Y\) are bona fide covariance matrices.  For any bona
fide Majorana covariance matrix \(C\), all 
singular values of \(C\) are at most
one.  Since the 
matrices here have size \(2m_A\times2m_A\),
\begin{equation}
        \trnorm{M_A+(n-1)Y}\le 2m_A,
        \qquad
        \trnorm{M_A-Y}\le 2m_A.
\end{equation}
Moreover,
\begin{equation}
        nY
        =
        \bigl(M_A+(n-1)Y\bigr)-(M_A-Y).
\end{equation}
Therefore, 
the triangle inequality gives
\begin{equation}
        n\trnorm{Y}
        \le
        \trnorm{M_A+(n-1)Y}
        +
        \trnorm{M_A-Y}
        \le
        4m_A.
\end{equation}
This proves Eq.~\eqref{eq:Y-trace-bound}.
\end{proof}

\begin{remark}[Sharpness of the \(1/n\) suppression]
For \(n\ge2\), the \(1/n\) scaling in
Lemma~\ref{lem:exchangeable-block-bound} is sharp
already at the one-mode covariance-matrix level.  Let
\begin{equation}
        J_{\rm f}
        =
        \begin{psmallmatrix}
        0 & 1\\
        -1 & 0
        \end{psmallmatrix},
\end{equation}
and consider a covariance matrix with
\begin{equation}
        Y=\frac{2J_{\rm f}}{n},
        \qquad
        M_A=-J_{\rm f}+\frac{2J_{\rm f}}{n}.
\end{equation}
Then the two sector covariance matrices appearing in
Proposition~\ref{prop:exchangeability-criterion} are
\begin{equation}
        M_A+(n-1)Y
        =
        J_{\rm f},
        \qquad
        M_A-Y
        =
        -J_{\rm f}.
\end{equation}
Both \(J_{\rm f}\) and \(-J_{\rm f}\) are admissible one-mode Majorana
covariance matrices.  Hence, by
Proposition~\ref{prop:exchangeability-criterion}, the corresponding
permutation-invariant \(n\)-copy covariance matrix is admissible.  Its common
off-diagonal block satisfies
\begin{equation}
        \trnorm{Y}
        =
        \frac{2}{n}\trnorm{J_{\rm f}}
        =
        \frac{4}{n}.
\end{equation}
For \(m_A=1\), Lemma~\ref{lem:exchangeable-block-bound} gives the upper bound
\(\trnorm{Y}\le4/n\).  Thus both the \(1/n\) scaling and the constant in the
lemma are attained by this one-mode example.
\end{remark}

\subsection{Proof of the de Finetti bound and infinite exchangeability}
\label{subsec:exchangeable-definetti-proof}

It remains to pass from covariance suppression to the trace-norm bound for
Gaussian states.  The product Gaussian state \(\rho_A^{\otimes k}\) has
covariance matrix \(M_A^{\oplus k}\), obtained from
\eqref{eq:exchangeable-block-matrix} by deleting all off-diagonal blocks.
Therefore
\begin{equation}
        M^{(k)}-M_A^{\oplus k}
        =
        \begin{pmatrix}
        0 & Y & \cdots & Y\\
        Y & 0 & \cdots & Y\\
        \vdots & \vdots & \ddots & \vdots\\
        Y & Y & \cdots & 0
        \end{pmatrix}.
\end{equation}
Let
\begin{equation}
        \ket{+_k}
        =
        \frac1{\sqrt k}(1,\ldots,1)^T,
        \qquad
        P_+^{(k)}=\ket{+_k}\!\bra{+_k},
        \qquad
        P_-^{(k)}=\one_k-P_+^{(k)} .
\end{equation}
The same sector decomposition gives
\begin{equation}
        M^{(k)}-M_A^{\oplus k}
        =
        (k-1)P_+^{(k)}\otimes Y
        -
        P_-^{(k)}\otimes Y .
\end{equation}
Hence, using the orthogonality of the two sectors and multiplicativity of the
trace norm under tensor products,
\begin{equation}
        \trnorm{M^{(k)}-M_A^{\oplus k}}
        =
        (k-1)\trnorm{Y}+(k-1)\trnorm{Y}
        =
        2(k-1)\trnorm{Y}.
\end{equation}
Together with Lemma~\ref{lem:exchangeable-block-bound}, this gives
\begin{equation}
        \trnorm{M^{(k)}-M_A^{\oplus k}}
        \le
        \frac{8m_A(k-1)}{n}.
        \label{eq:covariance-trace-distance}
\end{equation}

\begin{proof}[Proof of Theorem~\ref{thm:exchangeable-gaussian-definetti}]
We use the trace-norm bound for fermionic Gaussian states recalled in
\eqref{eq:gaussian-covariance-lipschitz-prelim}: if \(M\) and \(M'\) are valid
covariance matrices, then
\begin{equation}
        \trnorm{\rho_M-\rho_{M'}}
        \le
        \frac12\trnorm{M-M'}.
        \label{eq:gaussian-covariance-lipschitz}
\end{equation}
Applying \eqref{eq:gaussian-covariance-lipschitz} with
\(M'=M_A^{\oplus k}\), and using
\eqref{eq:covariance-trace-distance}, we obtain
\begin{equation}
        \trnorm{\rho_{A_1, \cdots, A_k}-\rho_A^{\otimes k}}
        \le
        \frac12\trnorm{M^{(k)}-M_A^{\oplus k}}
        \le
        \frac{4m_A(k-1)}{n}.
\end{equation}
The product state \(\rho_A^{\otimes k}\) is obtained from the de Finetti
integral in Eq.~\eqref{eq:finite-definetti-bound} by taking
\(\mu=\delta_{\rho_A}\).  This proves the theorem.
\end{proof}

The same bound also gives the infinite exchangeability statement.

\begin{corollary}[Infinite exchangeability implies product structure]
\label{cor:infinite-exchangeability-product}
Let \(k\ge2\), and let \(\rho_{A_1, \cdots , A_k}\) be a permutation-invariant fermionic Gaussian
state with covariance matrix \eqref{eq:exchangeable-block-matrix}.  If
\(\rho_{A_1, \cdots, A_k}\) is infinitely exchangeable, then \(Y=0\).  Hence
\begin{equation}
        \rho_{A_1, \cdots , A_k}
        =
        \rho_A^{\otimes k}.
        \label{eq:infinite-exchangeability-product}
\end{equation}
\end{corollary}

\begin{proof}
If the state is infinitely exchangeable, then
Lemma~\ref{lem:exchangeable-block-bound} applies for arbitrarily large \(n\).
Thus \(\trnorm{Y}\le4m_A/n\) for arbitrarily large \(n\), and hence \(Y=0\).
The covariance matrix is then \(M_A^{\oplus k}\), which is the covariance
matrix of the Gaussian product state \(\rho_A^{\otimes k}\).  Since fermionic
Gaussian states are determined by their covariance matrices, the state itself
is \(\rho_A^{\otimes k}\).
\end{proof}

\section{Clustering and Markov structure}
\label{sec:markov}

A basic locality property of quantum many-body states is clustering of
correlations.  In ordinary correlation clustering, one asks that connected
correlation functions of local observables decay as the supports of the
observables are separated.  In this section we use an information-theoretic
version of this idea, where the relevant quantity is the mutual information.
The advantage of mutual information is that it measures the total correlation
between two regions, rather than the correlation detected by a particular pair
of observables.

In order to proceed,
we equip \(\mathbb Z^\nu\) with the \(\ell_1\) metric and, for nonempty finite
regions, and set
\begin{equation}
        \dist(X,Y)
        :=
        \min_{x\in X,\,y\in Y}\|x-y\|_1.
        \label{eq:lattice-distance}
\end{equation}
We assume that every lattice site carries the same fixed finite number
\(m_{\rm site}\) of fermionic modes.

\begin{definition}[Exponential clustering of mutual information]
\label{def:mi-clustering}
Let \(\rho\) be a state on \(\mathbb Z^\nu\).  We say that \(\rho\) has
exponential clustering of mutual information if, for every pair of nonempty finite
reference regions \(X_0,Y_0\subset\mathbb Z^\nu\), there exist constants
\(c_{X_0,Y_0},\mu_{X_0,Y_0}>0\) such that, for all disjoint translates
\(X=x+X_0\) and \(Y=y+Y_0\),
\begin{equation}
        I_\rho(X:Y)
        \le
        c_{X_0,Y_0}e^{-\mu_{X_0,Y_0}\dist(X,Y)} .
        \label{eq:mi-clustering}
\end{equation}
Here
\begin{equation}
        I_\rho(X:Y)
        =
        S(\rho_X)+S(\rho_Y)-S(\rho_{X\cup Y}).
\end{equation}
This quantity is the mutual information.
\end{definition}

\begin{figure}[t]
\centering
\begin{tikzpicture}[
    scale=0.62,
    every node/.style={font=\small},
    site/.style={circle,fill=black,inner sep=0pt,minimum size=2.2pt},
    Afill/.style={
        fill=blue!30,
        draw=blue!70,
        thick,
        fill opacity=0.35,
        rounded corners=4pt
    },
    Cfill/.style={
        fill=red!30,
        draw=red!70,
        thick,
        fill opacity=0.35,
        rounded corners=4pt
    },
    Bfill/.style={
        fill=gray!15,
        draw=gray!50,
        thick
    }
]

\begin{scope}
\node[font=\normalsize] at (3.5,8.4) {(a) Mutual information};

\draw[thin,gray!60] (0,0) grid (7,7);
\foreach \x in {0,...,7}{
  \foreach \y in {0,...,7}{
    \node[site] at (\x,\y) {};
  }
}

\filldraw[Afill] (0.6,0.6) rectangle (3.4,3.4);

\filldraw[Cfill] (5.0,5.0) rectangle (6.7,6.7);

\node[blue!70!black,font=\large] at (2.0,2.0) {$A$};
\node[red!70!black,font=\large] at (5.85,5.85) {$C$};

\draw[<->,thick] (3.5,3.5) -- (4.95,4.95)
node[midway,above left=-1pt and -1pt] {$r$};

\node at (3.5,-1.0) {$I_\rho(A:C)$};
\end{scope}

\begin{scope}[xshift=10.8cm]
\node[font=\normalsize] at (3.5,8.4) {(b) Conditional mutual information};

\filldraw[Bfill] (0,0) rectangle (7,7);
\draw[thin,gray!60] (0,0) grid (7,7);

\foreach \x in {0,...,7}{
  \foreach \y in {0,...,7}{
    \node[site] at (\x,\y) {};
  }
}

\filldraw[Afill] (0.6,0.6) rectangle (3.4,3.4);
\filldraw[Cfill] (5.0,5.0) rectangle (6.7,6.7);

\node[blue!70!black,font=\large] at (2.0,2.0) {$A$};
\node[red!70!black,font=\large] at (5.85,5.85) {$C$};
\node[gray!70!black,font=\Large] at (1.0,5.8) {$B$};

\draw[<->,thick] (3.5,3.5) -- (4.95,4.95)
node[midway,above left=-1pt and -1pt] {$r$};

\node at (3.5,-1.0) {$I_\rho(A:C\mid B)$};
\end{scope}

\end{tikzpicture}
\caption{
A two-dimensional lattice illustrating mutual information and conditional
mutual information. In panel (a), \(A\) and \(C\) are two separated regions.
In panel (b), \(B\) denotes a finite buffer region between \(A\) and \(C\), so
that one considers the conditional mutual information \(I_\rho(A:C\mid B)\).
}
\label{fig:mi-cmi-simple}
\end{figure}
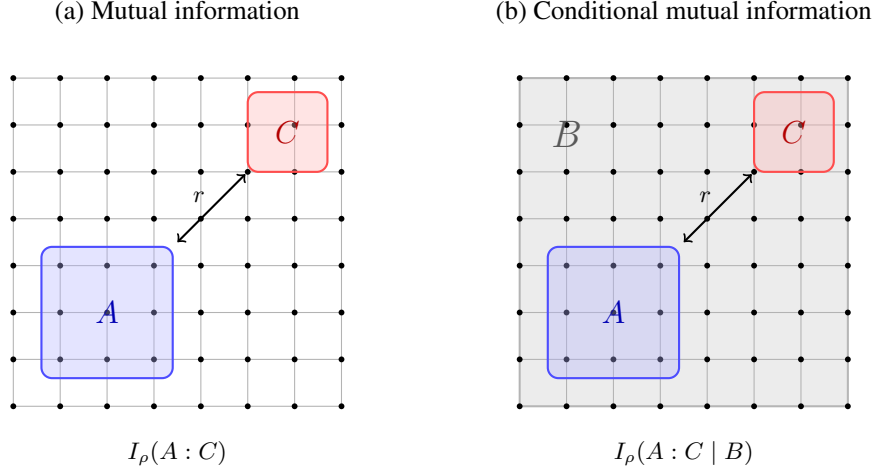

Thus the total correlation between two separated finite regions decays
exponentially with their separation.  This is a natural locality assumption in
quantum many-body theory.  It is stronger than asking for decay of a single
connected correlation function, since small mutual information controls all
connected correlations between bounded observables supported in the two
regions.  Exponential decay of connected correlations is standard for unique
gapped ground states of short-range Hamiltonians, but a spectral gap alone does
not
establish the mutual-information clustering assumed here
\citep{HastingsKoma,nachtergaele2005liebrobinsonboundsexponential}.
Exponential mutual-information clustering is known for several classes of
Gibbs states
\citep{bluhm2021exponentialdecaymutualinformation,
bluhm2024strongdecaycorrelations}.

We now consider a conditional version of the same locality question.  Let
\(A,B,C\) be pairwise disjoint finite regions, with \(B\) separating \(A\) from
\(C\).  The conditional mutual information is
\begin{equation}
        I_\rho(A:C\mid B)
        =
        S(\rho_{A\cup B})
        +
        S(\rho_{B\cup C})
        -
        S(\rho_B)
        -
        S(\rho_{A\cup B\cup C}) .
\end{equation}
Unlike \(I_\rho(A:C)\), this quantity does not measure the total marginal
correlation between \(A\) and \(C\).  Instead, it measures the part of the
correlation between \(A\) and \(C\) that remains after the degrees of freedom
in the intermediate region \(B\) have been retained.

This distinction is important.  In general quantum states, mutual information
and conditional mutual information are logically independent.  The mutual
information \(I_\rho(A:C)\) may be large while \(I_\rho(A:C\mid B)\) vanishes,
if all correlations between \(A\) and \(C\) are mediated through \(B\).
Conversely, \(I_\rho(A:C)\) may vanish while conditioning on \(B\) creates
correlations between \(A\) and \(C\).  Thus decay of ordinary mutual information
does not by itself imply decay of conditional mutual information for general
quantum states.

The significance of conditional mutual information is that it characterizes
quantum Markov structure.  The identity \(I_\rho(A:C\mid B)=0\) is equivalent
to exact recoverability of \(\rho_{A\cup B\cup C}\) from \(\rho_{A\cup B}\) by
a quantum channel acting only on \(B\) \citep{HaydenJozsaPetzWinter}.
Moreover, the recoverability theorem of Fawzi and Renner, together with its
universal recovery refinements, shows that small conditional mutual information
implies that \(\rho_{A\cup B\cup C}\) is close to a recovered state
\cite{FawziRenner,SutterTomamichelHarrow,
JungeRennerSutterWildeWinter,WildeRecoverability} 
\begin{equation}
        (\id_A\otimes\mathcal R_{B\to B\cup C})(\rho_{A\cup B}).
\end{equation}
In this sense,
\(I_\rho(A:C\mid B)\) is the natural quantitative measure of the obstruction to
recovery across the tripartition \(A|B|C\).

For pure tripartite states the distinction between the two quantities
collapses: if \(\rho_{A\cup B\cup C}\) is pure, then
\(I_\rho(A:C\mid B)=I_\rho(A:C)\).  Thus, in a pure-state setting,
mutual-information clustering immediately gives the corresponding conditional
statement.  For mixed states, and in particular for thermal states, this
identity is unavailable.  Conditional mutual information and approximate Markov
properties have therefore been studied separately for Gibbs states, including
one-dimensional Gibbs states, high-temperature Gibbs states, and more recent
finite-temperature settings
\citep{Intensive,KatoBrandao,
KuwaharaKatoBrandao,
kuwahara2024clusteringconditional,
kato2025clusteringconditionalmutual}.

The main result of this section is a conditional Markov bound for fermionic
Gaussian states with exponentially clustering mutual information.  In the
Gaussian setting, and for fixed finite outer regions, this state-level
clustering assumption implies exponential decay of the corresponding
off-diagonal Green-function blocks, after changing constants and decay rates.
Thus the proof below is formulated at the one-particle level, but no separate
one-particle locality assumption is introduced.

\begin{theorem}[Conditional mutual information from mutual-information clustering]
\label{thm:green-locality-cmi-decay}
Let \(\rho\) be a fermionic Gaussian state on \(\mathbb Z^\nu\) with
exponentially clustering mutual information in the sense of
Definition~\ref{def:mi-clustering}.  Fix nonempty finite reference regions
\(A_0\) and \(C_0\).  Let \(A=x+A_0\) and \(C=y+C_0\) be any disjoint
translates, and let \(B=B_r\) be any finite region such that \(A,B_r,C\) are
pairwise disjoint.  Write
\begin{equation}
        r=\dist(A,C)>0.
\end{equation}
Then there is a constant \(K_{A_0,C_0,\rho}>0\), independent of the translates
\(x,y\), of \(r\), and of the finite buffer \(B_r\), such that
\begin{equation}
        I_\rho(A:C\mid B_r)
        \le
        \frac{K_{A_0,C_0,\rho}}{r}.
        \label{eq:cmi-decay-main}
\end{equation}
Consequently, by recoverability, there exists a channel
\(\mathcal R_{B_r\to B_r\cup C}\) such that
\begin{equation}
        \left\|
        \rho_{A\cup B_r\cup C}
        -
        (\id_A\otimes\mathcal R_{B_r\to B_r\cup C})
        (\rho_{A\cup B_r})
        \right\|_1
        \le
        2\sqrt{I_\rho(A:C\mid B_r)}
        =O(r^{-1/2}).
\end{equation}
\end{theorem}

The proof is carried out at the covariance-matrix level.  For fermionic
Gaussian states, the entropy of a region is a trace function of the compressed
one-particle Green operator.  This allows the conditional mutual information to
be written as a series of inclusion--exclusion trace differences.  The
inclusion--exclusion structure cancels all closed one-particle words that do
not visit both \(A\) and \(C\).  The remaining words must travel from \(A\) to
\(C\) and back, and therefore inherit decay from the Green-function locality.
The short part of the series is controlled by this locality estimate, while the
long tail is controlled by a uniform rank bound.  Optimizing the cutoff gives
the \(O(r^{-1})\) bound in Theorem~\ref{thm:green-locality-cmi-decay}.

\subsection{Expansion of the conditional mutual information}
\label{subsec:cmi-expansion}

We first rewrite the conditional mutual information in a form suitable for
one-particle estimates.  Let \(\rho\) be a fermionic Gaussian state with
Majorana covariance matrix \(M\), and let
\begin{equation}
        G=\frac{\one+iM}{2},
\end{equation}
be the corresponding Majorana Green operator.  For a finite region \(X\), we
write
\begin{equation}
        G_X=\Pi_XG\Pi_X,
\end{equation}
for the compression of \(G\) to the Majorana one-particle space over \(X\).
Thus \(G_{A\cup B}\) denotes the compression to the union \(A\cup B\).  This is
distinct from the off-diagonal block notation \(G_{x, y}=\Pi_XG\Pi_Y\), which
will be used below only when \(X\) and \(Y\) are single lattice sites or
disjoint regions.

The reduced state \(\rho_X\) is Gaussian and is determined by \(G_X\).  Its
entropy is
\begin{equation}
        S(\rho_X)
        =
        -\tr\bigl(G_X\log G_X\bigr),
        \label{eq:entropy-green-region},
\end{equation}
where the trace is over the Majorana one-particle space of \(X\).

For \(k\ge1\), define the inclusion--exclusion trace difference
\begin{equation}
        \Delta_k
        :=
        \tr\bigl(
        G_{A\cup B}^k
        +
        G_{B\cup C}^k
        -
        G_B^k
        -
        G_{A\cup B\cup C}^k
        \bigr).
        \label{eq:Delta-k-definition}
\end{equation}
The following expansion is the starting point for the bound.

\begin{proposition}[Power-series expansion of the conditional mutual information]
\label{prop:cmi-series}
For a fermionic Gaussian state on \(A\cup B\cup C\),
\begin{equation}
        I_\rho(A:C\mid B)
        =
        -
        \sum_{k=2}^{\infty}
        \frac{\Delta_k}{k(k-1)} .
        \label{eq:cmi-series}
\end{equation}
\end{proposition}

\begin{proof}
We use the scalar identity
\begin{equation}
        (1-x)\log(1-x)
        =
        -x+\sum_{k=2}^{\infty}\frac{x^k}{k(k-1)},
        \qquad 0\le x\le1 .
        \label{eq:one-minus-x-log-expansion}
\end{equation}
The identity is understood by continuity at the endpoints \(x=0,1\).  Applying
it to \(G_X\) by functional calculus gives
\begin{equation}
        -\tr\bigl((\one_X-G_X)\log(\one_X-G_X)\bigr)
        =
        \tr G_X
        -
        \sum_{k=2}^{\infty}
        \frac{\tr(G_X^k)}{k(k-1)} .
\end{equation}
The particle-hole symmetry gives \(G_X^T=\one_X-G_X\), so \(G_X\) and
\(\one_X-G_X\) have the same spectrum.  Hence the left-hand side equals
\(S(\rho_X)\), and therefore
\begin{equation}
        S(\rho_X)
        =
        \tr G_X
        -
        \sum_{k=2}^{\infty}
        \frac{\tr(G_X^k)}{k(k-1)} .
        \label{eq:entropy-power-series}
\end{equation}

Substituting \eqref{eq:entropy-power-series} into the definition of
\(I_\rho(A:C\mid B)\), the linear terms cancel:
\begin{equation}
        \tr G_{A\cup B}
        +
        \tr G_{B\cup C}
        -
        \tr G_B
        -
        \tr G_{A\cup B\cup C}
        =
        0 .
\end{equation}
The remaining terms are exactly the inclusion--exclusion traces
\(\Delta_k\).  Thus
\begin{equation}
\begin{aligned}
        I_\rho(A:C\mid B)
        &=
        -
        \sum_{k=2}^{\infty}
        \frac{
        \tr(G_{A\cup B}^k)
        +
        \tr(G_{B\cup C}^k)
        -
        \tr(G_B^k)
        -
        \tr(G_{A\cup B\cup C}^k)}
        {k(k-1)}  \\
        &=
        -
        \sum_{k=2}^{\infty}
        \frac{\Delta_k}{k(k-1)} .
\end{aligned}
\end{equation}
This proves \eqref{eq:cmi-series}.
\end{proof}

The expansion \eqref{eq:cmi-series} reduces the problem to estimating
\(\Delta_k\).  The expression \(\Delta_k\) is designed so that contributions
which do not connect \(A\) to \(C\) cancel.  This cancellation will be used
for the locality bound.  Before using locality, however, we prove a uniform
bound which is independent of the separation between \(A\) and \(C\).

\subsection{Bounds on the trace differences}
\label{subsec:delta-estimates}

There are two bounds on \(\Delta_k\).  The first one is purely
finite-dimensional: it only uses the fact that \(A\) and \(C\) contain finitely
many one-particle modes.  This gives a uniform bound in \(k\), which will later
control the tail of \eqref{eq:cmi-series}.  The second bound uses spatial
locality of \(G\).  It gives decay in the distance between \(A\) and \(C\), but
has a constant which grows with \(k\).  This is why it will only be used below
a cutoff.

For a finite region \(X\), let
\begin{equation}
        N_X:=\rank\Pi_X,
\end{equation}
be the dimension of the Majorana one-particle space over \(X\).  If each
lattice site carries \(m_{\rm site}\) fermionic modes, then
\begin{equation}
        N_X=2m_{\rm site}|X| .
\end{equation}
We use \(N_X\), rather than \(M_X\), because \(M_X\) denotes the covariance
matrix of the region \(X\).

We begin with the elementary finite-dimensional bound used in the rank
estimate.

\begin{lemma}[Trace bound under projection]
\label{lem:projection-trace}
Let \(T\) be an \(N\times N\) Hermitian matrix with \(0\preceq T\preceq\one\),
and let \(P\) be an orthogonal projection.  Then, for every \(k\ge1\),
\begin{equation}
        0
        \le
        \tr T^k-\tr\bigl((PTP)^k\bigr)
        \le
        \rank(\one-P).
        \label{eq:projection-trace}
\end{equation}
\end{lemma}

\begin{proof}
Let \(L=\rank P\).  Choose an orthonormal basis in which
\(P=\operatorname{diag}(\one_L,0)\).  In this basis,
\begin{equation}
        PTP=
        \begin{pmatrix}
        T_P & 0\\
        0 & 0
        \end{pmatrix},
\end{equation}
where \(T_P\) is an \(L\times L\) principal submatrix of \(T\).  Let
\(\lambda_1\ge\cdots\ge\lambda_N\) be the eigenvalues of \(T\), and let
\(\mu_1\ge\cdots\ge\mu_L\) be the eigenvalues of \(T_P\).  By Cauchy's
interlacing theorem,
\begin{equation}
        \lambda_j\ge \mu_j\ge \lambda_{j+N-L},
        \qquad 1\le j\le L .
\end{equation}
Since all eigenvalues lie in \([0,1]\), the function \(x\mapsto x^k\) is
increasing on the relevant interval.  Hence
\begin{equation}
        \sum_{j=1}^L \lambda_j^k
        \ge
        \sum_{j=1}^L \mu_j^k
        \ge
        \sum_{j=1}^L \lambda_{j+N-L}^k .
\end{equation}
Because \(\tr((PTP)^k)=\sum_{j=1}^L\mu_j^k\), we get
\begin{equation}
        \tr T^k-\tr\bigl((PTP)^k\bigr)
        =
        \sum_{j=1}^N\lambda_j^k-\sum_{j=1}^L\mu_j^k
        \ge 0 .
\end{equation}
For the upper bound,
\begin{equation}
\begin{aligned}
        \tr T^k-\tr\bigl((PTP)^k\bigr)
        &\le
        \sum_{j=1}^N\lambda_j^k
        -
        \sum_{j=1}^L\lambda_{j+N-L}^k  \\
        &=
        \sum_{j=1}^{N-L}\lambda_j^k
        \le
        N-L
        =
        \rank(\one-P).
\end{aligned}
\end{equation}
This proves \eqref{eq:projection-trace}.
\end{proof}

\begin{lemma}[Uniform trace bound]
\label{lem:uniform-delta}
For every \(k\ge1\),
\begin{equation}
        |\Delta_k|
        \le
        \min\{N_A,N_C\}.
        \label{eq:uniform-delta}
\end{equation}
\end{lemma}

\begin{proof}
We first group the terms according to the region \(A\):
\begin{equation}
        \Delta_k
        =
        \bigl(\tr G_{A\cup B}^k-\tr G_B^k\bigr)
        -
        \bigl(\tr G_{A\cup B\cup C}^k-\tr G_{B\cup C}^k\bigr).
\end{equation}
Inside the one-particle space of \(A\cup B\), let \(P_B\) be the projection
onto the \(B\)-modes.  Then
\begin{equation}
        G_B=P_BG_{A\cup B}P_B .
\end{equation}
Since \(0\preceq G_{A\cup B}\preceq\one\) and
\(\rank(\one-P_B)=N_A\), Lemma~\ref{lem:projection-trace} gives
\begin{equation}
        0
        \le
        \tr G_{A\cup B}^k-\tr G_B^k
        \le
        N_A.
\end{equation}
Similarly, \(G_{B\cup C}\) is obtained from \(G_{A\cup B\cup C}\) by projecting
away the \(A\)-modes, and hence
\begin{equation}
        0
        \le
        \tr G_{A\cup B\cup C}^k-\tr G_{B\cup C}^k
        \le
        N_A .
\end{equation}
Thus \(\Delta_k\) is the difference of two numbers in \([0,N_A]\), and so
\begin{equation}
        |\Delta_k|\le N_A .
\end{equation}

Interchanging \(A\) and \(C\), we also write
\begin{equation}
        \Delta_k
        =
        \bigl(\tr G_{B\cup C}^k-\tr G_B^k\bigr)
        -
        \bigl(\tr G_{A\cup B\cup C}^k-\tr G_{A\cup B}^k\bigr).
\end{equation}
The same argument gives \(|\Delta_k|\le N_C\).  Combining the two bounds proves
\eqref{eq:uniform-delta}.
\end{proof}

The bound \eqref{eq:uniform-delta} does not use locality.  Inserted directly
into \eqref{eq:cmi-series}, it gives only the trivial bound
\(I_\rho(A:C\mid B)\le\min\{N_A,N_C\}\), since
\(\sum_{k\ge2}1/(k(k-1))=1\).  To obtain decay in the separation between
\(A\) and \(C\), we need a bound which uses locality of the Green operator.

\subsection{Locality bound for the finite part of the series}
\label{subsec:locality-estimate}

We now exploit mutual-information clustering to derive quantitative spatial decay directly at the level of the covariance matrix. Let $x,y\in\mathbb Z^\nu$ denote lattice sites, and let $G_{x,y}$ be the one-particle block of the Green operator mapping the Majorana space associated with $y$ to that associated with $x$.
Since each site carries only finitely many
fermionic modes, exponential clustering of mutual information implies, after
possibly changing constants, an exponential decay estimate for these blocks:
there exist constants \(C_0,\alpha>0\) such that
\begin{equation}
        \trnorm{G_{x, y}}
        \le
        C_0 e^{-\alpha |x-y|_1},
        \qquad x,y\in\mathbb Z^\nu .
        \label{eq:green-locality}
\end{equation}

Let us briefly recall why this follows from
Definition~\ref{def:mi-clustering}.  For single-site regions \(x\) and \(y\),
Pinsker's inequality gives a trace-distance bound on
\(\rho_{x, y}-\rho_x\otimes\rho_y\) in terms of \(I_\rho(x:y)^{1/2}\).  The
entries of the off-diagonal covariance block are expectation values of
bounded Majorana bilinears, and hence are controlled by this trace distance.
Since the internal one-site Majorana dimension is fixed, entrywise exponential
decay is equivalent, up to a constant, to trace-norm exponential decay of the
block.  Finally, for \(x\ne y\),
\begin{equation}
        G_{x, y}
        =
        \frac{i}{2}M_{x, y},
\end{equation}
so the same exponential decay holds for the Green-function block.

The next lemma says that this exponential decay is stable under taking powers.
We state it in the form needed below, uniformly for all finite compressions of
the Green operator.  Each multiplication costs one convolution factor, so the
constant grows exponentially in the power.  This growth is the reason the
resulting bound will only be used for \(k\) below a cutoff.

For \(0\le\beta<\alpha\), define
\begin{equation}
        Q_\beta
        :=
        C_0
        \sum_{z\in\mathbb Z^\nu}
        e^{-(\alpha-\beta)|z|_1}.
        \label{eq:Q-beta}
\end{equation}
The sum is finite for every \(\beta<\alpha\), and \(Q_\beta\ge C_0\).

\begin{lemma}[Decay of powers]
\label{lem:power-locality}
Let \(R\subset\mathbb Z^\nu\) be finite, and let
\(G_R=\Pi_RG\Pi_R\).  For every \(0\le\beta<\alpha\), every \(\ell\ge1\), and
all \(x,y\in R\),
\begin{equation}
        \trnorm{(G_R^\ell)_{x, y}}
        \le
        Q_\beta^\ell e^{-\beta |x-y|_1}.
        \label{eq:power-locality}
\end{equation}
\end{lemma}

\begin{proof}
We prove the bound by induction on \(\ell\).  For \(\ell=1\), the estimate is
\begin{equation}
        \trnorm{(G_R)_{x, y}}
        =
        \trnorm{G_{x, y}}
        \le
        C_0 e^{-\alpha |x-y|_1}
        \le
        Q_\beta e^{-\beta |x-y|_1},
\end{equation}
using \(Q_\beta\ge C_0\) and \(\beta<\alpha\).

Assume the bound holds for some \(\ell\ge1\).  Since
\begin{equation}
        (G_R^{\ell+1})_{x, y}
        =
        \sum_{z\in R}(G_R^\ell)_{xz}G_{zy},
\end{equation}
the induction hypothesis, \eqref{eq:green-locality}, and
\(|x-y|_1\le |x-z|_1+|z-y|_1\) give
\begin{equation}
\begin{aligned}
        \trnorm{(G_R^{\ell+1})_{x, y}}
        &\le
        \sum_{z\in R}
        \trnorm{(G_R^\ell)_{xz}}\,
        \trnorm{G_{zy}}  \\
        &\le
        Q_\beta^\ell C_0
        \sum_{z\in R}
        e^{-\beta |x-z|_1}
        e^{-\alpha |z-y|_1}  \\
        &\le
        Q_\beta^\ell C_0
        e^{-\beta |x-y|_1}
        \sum_{z\in R}
        e^{-(\alpha-\beta)|z-y|_1}  \\
        &\le
        Q_\beta^\ell C_0
        e^{-\beta |x-y|_1}
        \sum_{z\in\mathbb Z^\nu}
        e^{-(\alpha-\beta)|z-y|_1}.
\end{aligned}
\end{equation}
By translation invariance of the lattice sum, the last sum is \(Q_\beta/C_0\).
Hence
\begin{equation}
        \trnorm{(G_R^{\ell+1})_{x, y}}
        \le
        Q_\beta^{\ell+1}e^{-\beta |x-y|_1}.
\end{equation}
This completes the induction.
\end{proof}

The previous lemma bounds a path of length \(\ell\) between two fixed sites. The next estimate uses the inclusion--exclusion structure of \(\Delta_k\).
For this purpose it is useful to work with a positive path-norm quantity.
Let \(R\subset\mathbb Z^\nu\) be finite.  For \(x,y\in R\) and
\(\ell\ge1\), define
\begin{equation}
        L_\ell^R(x,y)
        :=
        \sum_{\substack{z_0,\ldots,z_\ell\in R\\
        z_0=x,\ z_\ell=y}}
        \prod_{j=0}^{\ell-1}
        \trnorm{G_{z_jz_{j+1}}}.
        \label{eq:path-norm-definition}
\end{equation}
Thus \(L_1^R(x,y)=\trnorm{G_{x, y}}\).  Equivalently, if
\(\mathsf W_R\) is the scalar nonnegative matrix with entries
\begin{equation}
        (\mathsf W_R)_{x, y}:=\trnorm{G_{x, y}},
        \qquad x,y\in R,
\end{equation}
then
\begin{equation}
        L_\ell^R(x,y)=(\mathsf W_R^\ell)_{x, y}.
\end{equation}

Let
\begin{equation}
        R:=A\cup B\cup C,
        \qquad
        r=\dist(A,C).
\end{equation}

\begin{lemma}[Cycle-splitting bound]
\label{lem:cycle-splitting}
For every \(k\ge2\),
\begin{equation}
        |\Delta_k|
        \le
        2k
        \sum_{m=1}^{k-1}
        \sum_{a\in A}
        \sum_{c\in C}
        L_m^R(a,c)\,
        L_{k-m}^R(c,a).
        \label{eq:cycle-splitting}
\end{equation}
\end{lemma}

\begin{proof}
We expand the four traces in the definition of \(\Delta_k\) as sums over
closed one-particle words of length \(k\).  A closed word is a sequence
\begin{equation}
        \mathbf x=(x_1,\ldots,x_k)\in R^k,
        \qquad x_{k+1}:=x_1,
\end{equation}
and its contribution to \(\tr(G_R^k)\) is
\begin{equation}
        T(\mathbf x)
        :=
        \tr\!\left(
        G_{x_1x_2}G_{x_2x_3}\cdots G_{x_kx_1}
        \right).
\end{equation}
The inclusion--exclusion coefficient of this word in \(\Delta_k\) is
\begin{equation}
\begin{aligned}
        \chi(\mathbf x)
        &:=
        \one_{\{x_1,\ldots,x_k\subset A\cup B\}}
        +
        \one_{\{x_1,\ldots,x_k\subset B\cup C\}}  \\
        &\quad
        -
        \one_{\{x_1,\ldots,x_k\subset B\}}
        -
        \one_{\{x_1,\ldots,x_k\subset R\}} .
\end{aligned}
\end{equation}
Since the word is always contained in \(R\), the last indicator is equal to
one.  If the word is contained in \(B\), then the four terms give
\(1+1-1-1=0\).  If the word visits \(A\) but not \(C\), then it is contained
in \(A\cup B\) but not in \(B\cup C\) or \(B\), and the coefficient is
\(1+0-0-1=0\).  Similarly, if the word visits \(C\) but not \(A\), the
coefficient is also zero.  Therefore only words which visit both \(A\) and
\(C\) contribute.  For such words,
\begin{equation}
        \chi(\mathbf x)=-1,
\end{equation}
and hence
\begin{equation}
        \Delta_k
        =
        -
        \sum_{\substack{\mathbf x\in R^k:\\
        \{x_1,\ldots,x_k\}\cap A\ne\emptyset,\,
        \{x_1,\ldots,x_k\}\cap C\ne\emptyset}}
        T(\mathbf x).
        \label{eq:mixed-word-expansion}
\end{equation}

For every closed word,
\begin{equation}
\begin{aligned}
        |T(\mathbf x)|
        &=
        \left|
        \tr\!\left(
        G_{x_1x_2}G_{x_2x_3}\cdots G_{x_kx_1}
        \right)
        \right|  \\
        &\le
        \prod_{j=1}^k \trnorm{G_{x_jx_{j+1}}}.
\end{aligned}
\end{equation}
Indeed, this follows from
\begin{equation}
        |\tr(Y_1\cdots Y_k)|
        \le
        \trnorm{Y_1}\prod_{j=2}^k\|Y_j\|_\infty
        \le
        \prod_{j=1}^k\trnorm{Y_j}.
\end{equation}

We now split each contributing closed word at a visit to \(A\) and a visit to
\(C\).  Since the word is cyclic, we may choose a starting point in \(k\)
ways.  If \(x_j\in A\) and \(x_{j+m}\in C\), with cyclic indices and
\(1\le m\le k-1\), then the word is split into a path of length \(m\) from
\(A\) to \(C\), and a path of length \(k-m\) from \(C\) back to \(A\).
Allowing both cyclic orientations gives
\begin{equation}
\begin{aligned}
        |\Delta_k|
        &\le
        2
        \sum_{j=1}^k
        \sum_{m=1}^{k-1}
        \sum_{a\in A}
        \sum_{c\in C}
        \sum_{\substack{\mathbf x\in R^k:\\
        x_j=a,\ x_{j+m}=c}}
        \prod_{\ell=1}^k
        \trnorm{G_{x_\ell x_{\ell+1}}}.
\end{aligned}
\end{equation}
For fixed \(j,m,a,c\), the sum over the intermediate sites on the first piece
of the cycle gives \(L_m^R(a,c)\), while the sum over the intermediate sites
on the second piece gives \(L_{k-m}^R(c,a)\).  Therefore
\begin{equation}
        |\Delta_k|
        \le
        2k
        \sum_{m=1}^{k-1}
        \sum_{a\in A}
        \sum_{c\in C}
        L_m^R(a,c)\,
        L_{k-m}^R(c,a).
\end{equation}
This proves \eqref{eq:cycle-splitting}.
\end{proof}

Combining the cycle decomposition with the decay of the one-particle blocks
gives the desired locality bound for \(\Delta_k\).

\begin{lemma}[Locality bound for mixed words]
\label{lem:mixed-delta}
Let \(0<\beta<\alpha\).  For every \(k\ge2\),
\begin{equation}
        |\Delta_k|
        \le
        \frac{2kQ_0}{Q_\beta-Q_0}\,
        |A|\,|C|\,Q_\beta^k e^{-\beta r}.
        \label{eq:mixed-delta}
\end{equation}
Here \(|A|\) and \(|C|\) denote the numbers of lattice sites in the two
regions.
\end{lemma}

\begin{proof}
We first prove the corresponding decay estimate for the path-norm quantities.
For every \(0\le\gamma<\alpha\), every \(\ell\ge1\), and all \(x,y\in R\),
\begin{equation}
        L_\ell^R(x,y)
        \le
        Q_\gamma^\ell e^{-\gamma |x-y|_1}.
        \label{eq:path-norm-locality}
\end{equation}
For \(\ell=1\), this follows from \eqref{eq:green-locality}:
\begin{equation}
        L_1^R(x,y)
        =
        \trnorm{G_{x, y}}
        \le
        C_0 e^{-\alpha |x-y|_1}
        \le
        Q_\gamma e^{-\gamma |x-y|_1},
\end{equation}
because \(Q_\gamma\ge C_0\).

Assume now that \eqref{eq:path-norm-locality} holds for some \(\ell\ge1\).
By the definition of \(L_{\ell+1}^R\),
\begin{equation}
        L_{\ell+1}^R(x,y)
        =
        \sum_{z\in R}
        L_\ell^R(x,z)\,\trnorm{G_{zy}}.
\end{equation}
Using the induction hypothesis, \eqref{eq:green-locality}, and the triangle
inequality
\begin{equation}
        |x-y|_1\le |x-z|_1+|z-y|_1,
\end{equation}
we get
\begin{equation}
\begin{aligned}
        L_{\ell+1}^R(x,y)
        &\le
        Q_\gamma^\ell C_0
        \sum_{z\in R}
        e^{-\gamma |x-z|_1}
        e^{-\alpha |z-y|_1}  \\
        &\le
        Q_\gamma^\ell C_0
        e^{-\gamma |x-y|_1}
        \sum_{z\in R}
        e^{-(\alpha-\gamma)|z-y|_1}  \\
        &\le
        Q_\gamma^\ell C_0
        e^{-\gamma |x-y|_1}
        \sum_{z\in\mathbb Z^\nu}
        e^{-(\alpha-\gamma)|z-y|_1}.
\end{aligned}
\end{equation}
By translation invariance of the lattice sum, the last sum is \(Q_\gamma/C_0\).
Thus
\begin{equation}
        L_{\ell+1}^R(x,y)
        \le
        Q_\gamma^{\ell+1}e^{-\gamma |x-y|_1}.
\end{equation}
This proves \eqref{eq:path-norm-locality} by induction.

We now apply Lemma~\ref{lem:cycle-splitting}.  Since
\begin{equation}
        |a-c|_1\ge r,
\end{equation}
for all \(a\in A\) and \(c\in C\), the estimates
\eqref{eq:path-norm-locality} with \(\gamma=0\) and \(\gamma=\beta\) give
\begin{equation}
        L_m^R(a,c)\,L_{k-m}^R(c,a)
        \le
        Q_0^m Q_\beta^{k-m}e^{-\beta r}.
\end{equation}
Substituting this into \eqref{eq:cycle-splitting}, we obtain
\begin{equation}
\begin{aligned}
        |\Delta_k|
        &\le
        2k|A|\,|C|\,e^{-\beta r}
        \sum_{m=1}^{k-1}
        Q_0^m Q_\beta^{k-m}  \\
        &=
        2k|A|\,|C|\,e^{-\beta r}Q_\beta^k
        \sum_{m=1}^{k-1}
        \left(\frac{Q_0}{Q_\beta}\right)^m .
\end{aligned}
\end{equation}
Since \(Q_\beta>Q_0\) for \(\beta>0\), the geometric sum is bounded by
\begin{equation}
        \sum_{m=1}^{k-1}
        \left(\frac{Q_0}{Q_\beta}\right)^m
        \le
        \frac{Q_0}{Q_\beta-Q_0}.
\end{equation}
Therefore
\begin{equation}
        |\Delta_k|
        \le
        \frac{2kQ_0}{Q_\beta-Q_0}
        |A|\,|C|\,Q_\beta^k e^{-\beta r}.
\end{equation}
This proves \eqref{eq:mixed-delta}.
\end{proof}

The bound \eqref{eq:mixed-delta} contains the desired spatial factor
\(e^{-\beta r}\), but it also contains the growth factor \(Q_\beta^k\).
Therefore it is useful only for the finite part of the series.  The large
powers will be controlled by the uniform bound \eqref{eq:uniform-delta}.

\subsection{Decay of the conditional mutual information}
\label{subsec:cmi-decay-proof}

We now combine the two bounds.  The proof is a cutoff argument: below the
cutoff we use \eqref{eq:mixed-delta}, which sees the separation \(r\); above
the cutoff we use \eqref{eq:uniform-delta}, which is independent of locality
but summable in \(k\).  Since \(Q_\beta\to\infty\) as
\(\beta\uparrow\alpha\), we may choose \(0<\beta<\alpha\) such that
\(Q_\beta>1\).

\begin{theorem}[Quantitative decay of conditional mutual information]
\label{thm:cmi-decay-quantitative}
Let \(\rho\) be a fermionic Gaussian state with exponentially clustering
mutual information.  Use the resulting Green-function clustering bound
\eqref{eq:green-locality}.  Fix nonempty finite reference regions \(A_0\) and
\(C_0\), let \(A\) and \(C\) range over disjoint translates of \(A_0\) and
\(C_0\), and let \(B\) be any finite region such that \(A,B,C\) are pairwise
disjoint.  Set \(r=\dist(A,C)>0\).  Then, uniformly over these choices,
\begin{equation}
        I_\rho(A:C\mid B)=O(r^{-1}),
        \qquad r\to\infty .
        \label{eq:cmi-decay-final}
\end{equation}
More explicitly, for every \(0<\beta<\alpha\) with \(Q_\beta>1\), and all
\(r\ge4\beta^{-1}\log Q_\beta\),
\begin{equation}
        I_\rho(A:C\mid B)
        \le
        C_\beta |A|\,|C|\,e^{-\beta r/2}
        +
        \frac{8\log Q_\beta}{\beta r}\min\{N_A,N_C\},
        \label{eq:cmi-explicit-bound},
\end{equation}
where
\begin{equation}
        C_\beta
        :=
        \frac{2Q_0}{(Q_\beta-Q_0)(Q_\beta-1)} .
        \label{eq:C-beta-definition}
\end{equation}
\end{theorem}

\begin{proof}
From \eqref{eq:cmi-series},
\begin{equation}
        I_\rho(A:C\mid B)
        \le
        \sum_{k=2}^{\infty}
        \frac{|\Delta_k|}{k(k-1)} .
\end{equation}
Fix \(0<\beta<\alpha\) with \(Q_\beta>1\), and set
\begin{equation}
        K
        :=
        \left\lfloor
        \frac{\beta r}{4\log Q_\beta}
        \right\rfloor .
\end{equation}
For \(r\ge4\beta^{-1}\log Q_\beta\), we have \(K\ge1\).  We split the series
at \(K\).

For the finite part, Lemma~\ref{lem:mixed-delta} gives
\begin{equation}
\begin{aligned}
        \sum_{k=2}^{K}
        \frac{|\Delta_k|}{k(k-1)}
        &\le
        \frac{2Q_0}{Q_\beta-Q_0}
        |A|\,|C|\,e^{-\beta r}
        \sum_{k=2}^{K}\frac{Q_\beta^k}{k-1}  \\
        &\le
        \frac{2Q_0}{(Q_\beta-Q_0)(Q_\beta-1)}
        |A|\,|C|\,Q_\beta^{K+1}e^{-\beta r}.
\end{aligned}
\end{equation}
Since \(K+1\le \beta r/(2\log Q_\beta)\) under the present assumption on
\(r\), it follows that
\begin{equation}
        Q_\beta^{K+1}e^{-\beta r}
        \le
        e^{-\beta r/2}.
\end{equation}
Hence
\begin{equation}
        \sum_{k=2}^{K}
        \frac{|\Delta_k|}{k(k-1)}
        \le
        C_\beta |A|\,|C|\,e^{-\beta r/2},
\end{equation}
with \(C_\beta\) as in \eqref{eq:C-beta-definition}.

For the tail, Lemma~\ref{lem:uniform-delta} gives
\begin{equation}
        \sum_{k=K+1}^{\infty}
        \frac{|\Delta_k|}{k(k-1)}
        \le
        \min\{N_A,N_C\}
        \sum_{k=K+1}^{\infty}\frac1{k(k-1)}
        =
        \frac{\min\{N_A,N_C\}}{K}.
\end{equation}
Moreover, the same choice of \(K\) gives
\begin{equation}
        K\ge \frac{\beta r}{8\log Q_\beta}.
\end{equation}
Therefore
\begin{equation}
        \sum_{k=K+1}^{\infty}
        \frac{|\Delta_k|}{k(k-1)}
        \le
        \frac{8\log Q_\beta}{\beta r}\min\{N_A,N_C\}.
\end{equation}
Combining the finite part and the tail proves
\eqref{eq:cmi-explicit-bound}.  Since \(A\) and \(C\) have fixed finite shape,
the right-hand side is \(O(r^{-1})\) as \(r\to\infty\).
\end{proof}

\begin{proof}[Proof of Theorem~\ref{thm:green-locality-cmi-decay}]
Theorem~\ref{thm:cmi-decay-quantitative} gives
\begin{equation}
        I_\rho(A:C\mid B_r)
        \le
        C_\beta |A|\,|C|\,e^{-\beta r/2}
        +
        \frac{8\log Q_\beta}{\beta r}\min\{N_A,N_C\},
\end{equation}
for any fixed \(0<\beta<\alpha\) with \(Q_\beta>1\), once \(r\) is sufficiently
large.  Since \(A\) and \(C\) have fixed finite shape, the first term is
exponentially small and the second is \(O(r^{-1})\).  Enlarging the constant to
cover the finitely many smaller values of \(r\) gives
\begin{equation}
        I_\rho(A:C\mid B_r)
        \le
        \frac{K_{A_0,C_0,\rho}}{r}.
\end{equation}

The final recoverability statement follows from the Fawzi--Renner
recoverability theorem and its refinements: small conditional mutual
information implies the existence of a channel
\(\mathcal R_{B_r\to B_r\cup C}\) such that
\begin{equation}
        \left\|
        \rho_{A\cup B_r\cup C}
        -
        (\id_A\otimes\mathcal R_{B_r\to B_r\cup C})(\rho_{A\cup B_r})
        \right\|_1
        \le
        2\sqrt{I_\rho(A:C\mid B_r)}.
\end{equation}
The bound above therefore gives an
\(O(r^{-1/2})\) trace-norm approximate-recovery bound.
\end{proof}

\section{Discussion and outlook}\label{sec:discussion}

In this work, we have set out to develop a unified theory of approximation for
fermionic Gaussian systems directly at the level of covariance matrices.  Our
central observation is that the elementary admissibility constraint on Majorana
covariance matrices encodes a powerful monogamy principle: a fixed fermionic
region has only a finite covariance budget that it can distribute among other
regions.  Depending on the physical structure of the problem, this single
principle takes three different forms.  High connectivity dilutes correlations
across interaction edges, exchangeability suppresses correlations between
copies, and spatial separation eliminates the covariance paths connecting
distant regions.  These mechanisms lead, respectively, to product-state
approximations for high-degree quadratic Hamiltonians, a finite fermionic
Gaussian de Finetti theorem, and approximate-recovery guarantees for spatially
clustering Gaussian states.  The results thus identify covariance monogamy as a
common structural origin of mean-field behaviour, exchangeable product
structure, and approximate quantum Markovianity in fermionic Gaussian systems.

One point worth emphasizing is that the arguments are genuinely one-particle
arguments.  The relevant estimates are proved first at the level of covariance
matrices and are then converted into statements about energy, free energy,
trace norm, or conditional mutual information.  This is different from the
usual dimension-dependent finite-dimensional approach, in which one works
directly with many-body states and then uses general information-theoretic
monogamy or de Finetti estimates.  For fermionic Gaussian states, the
covariance matrix contains the full state information, and the fermionic
admissibility condition already encodes strong restrictions on how correlations
can be distributed.

The high-degree result illustrates this point in a variational setting.  For a
quadratic Hamiltonian normalized as an average over edges, the energy
contribution of an edge depends only on the covariance block across that edge.
Thus, on a high-degree graph, the covariance shared by a site with its many
neighbours must be small on average.  This gives a direct product-state
approximation to the ground-state energy, and the same argument applies to the
finite-temperature free energy by using the product of the one-site marginals
of the Gibbs state.  The result is tailored to the quadratic structure: it does
not claim that interacting fermionic Hamiltonians on high-degree graphs are
controlled by product states in the same way.  It would be interesting to
understand how stable this covariance-based argument is under perturbations by
weak or structured non-quadratic interactions.
In this sense, the result sharpens the high-connectivity mechanism of
Ref.~\citep{KrumnowZimborasEisertGroundApprox} in the exactly quadratic regime,
while the earlier result remains the more general statement for interacting
fermionic Hamiltonians.

The de Finetti theorem obtained here is also a specifically Gaussian statement.
In the usual quantum de Finetti theorem, an exchangeable state is approximated
by a mixture of tensor powers.  In the fermionic Gaussian setting considered in
this work, the conclusion is stronger in form: the approximating de Finetti
measure may be chosen as the point mass at the one-copy marginal.  The reason
is again visible at the
covariance level.  Permutation invariance forces the off-diagonal copy-copy
covariance blocks to be identical, while admissibility of the full extension
forces this common block to vanish at order \(n^{-1}\).  Hence an infinitely
exchangeable fermionic Gaussian state is exactly a product state.  This should
be viewed as a Gaussian exchangeability result rather than as a general
fermionic de Finetti theorem.  It complements our previous analysis of
extendibility for fermionic Gaussian states \citep{NegariSalekExtendibility},
where the focus was on bipartite extendibility and approximation by fermionic
separable states.

The Markov part of this work gives another example of the same philosophy.
For Gaussian states, reduced entropies are functions of compressed
one-particle Green operators.  This makes it possible to expand conditional
mutual information in terms of powers of one-particle operators and to estimate
the contribution of paths crossing a separating buffer.  Under exponential
clustering of mutual information, this yields a polynomial decay bound for the
conditional mutual information across the buffer and, by recoverability, a
trace-norm approximation by a recovered state.  The proof does not
suggest that the obtained rate is necessarily optimal.  In particular, it would
be natural to ask whether additional assumptions, such as a spectral gap,
translation invariance, or stronger one-particle locality, can improve the
decay of \(I(A:C\mid B_r)\) beyond the bound proved here.

There are several possible directions in which the results could be sharpened.
First, the covariance monogamy inequality is sharp, as
Example~\ref{ex:sharp} shows.  It remains open, however, whether the constants
in the resulting high-degree energy and free-energy estimates are optimal,
since their derivation also uses edge averaging, Cauchy--Schwarz, and only the
operator-norm normalization of the interactions.  Refined estimates may be
possible when additional structure is available, for example on bipartite
graphs, in translation-invariant systems, or for regions with non-uniform
numbers of modes.

Second, it would be useful to clarify to what extent the Gaussian assumptions
can be relaxed.  The covariance monogamy estimate itself only uses positivity
of fermionic covariance matrices and therefore applies to arbitrary fermionic
states.  The later consequences, however, use Gaussianity in essential ways:
quadratic Hamiltonians are evaluated by two-point data, Gaussian states are
determined by their covariance matrices, and Gaussian entropies are trace
functions of one-particle operators.  Extending any of these approximation
statements to non-Gaussian states would require controlling higher-order
correlations, either perturbatively or through additional physical assumptions.

Third, the de Finetti result raises natural questions about other fermionic
notions of exchangeability.  In this work the state being extended is
Gaussian, and the extension is invariant under graded fermionic copy
permutations.  One could
instead consider exchangeability under more general fermionic symmetries, or
ask for de Finetti-type conclusions when only certain marginals are Gaussian.
Another related question is whether covariance-level extendibility conditions
can be combined with quantitative estimates on non-Gaussian cumulants to give
approximate product or separability statements beyond the Gaussian class.

Finally, the conditional mutual information estimate suggests a broader
connection between one-particle locality and approximate quantum Markov
structure.  For Gaussian states, the one-particle Green operator provides a
bridge between spatial decay and entropic quantities.  It would be interesting
to determine whether similar mechanisms exist for other exactly solvable
classes of fermionic states, or for interacting systems that are close to
Gaussian in an appropriate sense.  Such extensions could help relate
recoverability, locality of correlations, and effective low-complexity
descriptions in fermionic many-body systems.

Overall, the results support the view that fermionic covariance matrices are
not only a convenient parametrization of Gaussian states, but also a useful
source of quantitative approximation principles.  High degree, exchangeability,
and spatial separation are different structural assumptions, but in the
fermionic Gaussian setting they all act by suppressing the same object:
off-diagonal covariance.  This common mechanism is what underlies the product
and de Finetti approximations and the recoverability estimates proved in this
work.

\section*{Acknowledgements}

We are grateful to Timothy Hsieh for useful discussions.
This research was supported in part by the Perimeter Institute for Theoretical
Physics, which is funded by the Government of Canada through the Department of
Innovation, Science and Economic Development, and by the Province of Ontario
through the Ministry of Colleges and Universities.  A.-R.N. acknowledges support
from the Natural Sciences and Engineering Research Council of Canada (NSERC)
under Discovery Grant No.~RGPIN-2018-04380, as well as from an Ontario Early
Researcher Award.  F.S.'s work is supported by the European Commission as a
Marie Sk{\l}odowska-Curie Global Fellow, and he gratefully acknowledges the
hospitality of the Institute for Quantum Computing (IQC), University of
Waterloo, and the Freie Universit{\"a}t Berlin. J.~E.\ has been supported by the BMFTR (MuniQC-Atoms), the Clusters of Excellence (ML4Q, MATH+), the QuantERA (SDPCode), the Munich Quantum Valley, Berlin Quantum, the Quantum Flagship (Millenion, Pasquans2), the DFG (CRC 183, SPP 2514), and the European Research Council (DebuQC).

\bibliography{fermionic_references_corrected,BigReferences70}

\end{document}